\documentclass[11pt]{article}%

\usepackage[mathscr]{euscript}

\usepackage{todonotes}

\usepackage{amsmath}
\usepackage{amssymb}
\usepackage{relsize}
\usepackage[margin=1in]{geometry}
\usepackage{paralist}
\usepackage{graphicx}   
\usepackage{caption}
\usepackage{color}
\usepackage{xspace}

\usepackage{newtxtext}

\usepackage{hyperref}%
\hypersetup{%
   breaklinks,%
   ocgcolorlinks,
   colorlinks=true,%
   linkcolor=[rgb]{0.45,0.0,0.0},%
   citecolor=[rgb]{0,0,0.45}
}

\usepackage{amsmath}%
\usepackage[amsmath,amsthm,thmmarks]{ntheorem}

\usepackage{titlesec}
\titlespacing*{\section}{0pt}{0.7\baselineskip}{0.5\baselineskip}
\titlespacing*{\subsection}{0pt}{0.7\baselineskip}{0.7\baselineskip}

\theoremseparator{.}%

\usepackage{titlesec}
\titlelabel{\thetitle. }

\numberwithin{figure}{section}
\numberwithin{table}{section}
\numberwithin{equation}{section}

\allowdisplaybreaks

\newtheorem{theorem}{Theorem}[section] 
\newtheorem{lemma}[theorem]{Lemma}

\newtheorem{claim}[theorem]{Claim}
\newtheorem{remark}[theorem]{Remark}

{\theorembodyfont{\rm}%
   \newtheorem{defn}[theorem]{Definition}%
   \newtheorem{observation}[theorem]{Observation}
   
}

\newcommand{\MakeBig} {\rule[-.2cm]{0cm}{0.4cm}}

\newlength{\savedparindent}

\providecommand{\pbrcx}[1]{\left[ {#1} \right]}

\renewcommand{\Re}{{\rm I\!\hspace{-0.025em} R}}

\newcommand{\Eqlab}[1]{\label{equation:#1}}
\newcommand{\Eqref}[1]{Eq.~(\ref{equation:#1})}

\newcommand{\lemlab}[1]{\label{lemma:#1}}
\newcommand{\lemref}[1]{Lemma~\ref{lemma:#1}}

\newcommand{\obslab}[1]{\label{observation:#1}}
\newcommand{\obsref}[1]{Observation~\ref{observation:#1}}

\newcommand{\claimlab}[1]{\label{claim:#1}}

\newcommand{\seclab}[1]{\label{section:#1}}
\newcommand{\secref}[1]{Section~\ref{section:#1}}
\newcommand{\thmlab}[1]{\label{theorem:#1}}
\newcommand{\thmref}[1]{Theorem~\ref{theorem:#1}}

\newcommand{\pth}[2][\!]{#1\left({#2}\right)}%

\newcommand{\brc}[1]{\left\{ {#1} \right\}}%
\newcommand{\sep}[1]{\,\left|\, {#1} \MakeBig\right.}%

\newcommand{\Vector}[1]{\mathop{\mathbf{Vec}}\pth{#1}}

\newcommand{\Tr}[1]{\mathop{\mathbf{Tr}}\pth{#1}}

\newcommand{\Id}{\mathbf{I}}

\newcommand{\iprod}[2]{\langle #1,#2 \rangle}
\newcommand{\norm}[2]{\left\lVert {#2} \right\rVert_{#1}}

\newcommand{\Ex}[1]{\mathop{\mathbf{E}}{\pbrcx{#1}}}
\newcommand{\PEx}[1]{\mathop{\mathbf{\widetilde{E}}}{\pbrcx{#1}}}
\newcommand{\PE}{\mathop{\mathbf{\widetilde{E}}}}

\newcommand{\Var}[1]{\mathop{\mathbf{Var}}{\pbrcx{#1}}}
\newcommand{\Prob}[1]{\mathop{\mathbf{Pr}}\!{\pbrcx{#1}}}

\newcommand{\cardin}[1]{\left| {#1} \right|}%

\newcommand{\Sym}{\mathbb{S}}

\newcommand{\orbit}[1]{\mathscr{O}{\,\pth{#1}}}

\DeclareMathOperator*{\Moplus}{\text{\raisebox{0.25ex}{\scalebox{0.8}{$\bigoplus$}}}}

\newcommand{\multichoose}[2]{\ensuremath{\left(\kern-.3em\left(\genfrac{}{}{0pt}{}{#1}{#2}\right)\kern-.3em\right)}}

\newcommand{\numdist}[1]{\mathop{\mathbf{\#}}{\,\pth{#1}}}

\newcommand{\multiset}[1]{\mathop{\mathbf{multiset}}{\pth{#1}}}
\newcommand{\atomize}[1]{\mathop{\mathbf{atomize}}{\pth{#1}}}
\newcommand{\invatomize}[1]{\mathop{\mathbf{atomize^{-1}}}{\pth{#1}}}
\newcommand{\interleave}[2]{\mathop{\mathbf{intrlv}}{\pth{#1,#2}}}

\newcommand{\survive}[1]{\mathcal{S}{\,\pth{#1}}}

\newcommand{\survseqp}[2]{\mathcal{I}_{#1}\pth{#2}}

\newcommand{\matA}[1]{\mathrm{A}{\pbrcx{#1}}}

\newcommand{\tenA}[1]{\mathcal{A}{\pbrcx{#1}}}

\newcommand{\matB}[1]{\mathrm{B}{\pbrcx{#1}}}

\newcommand{\matT}[1]{\mathrm{T}{\pbrcx{#1}}}

\def\restrict#1{\raise-.5ex\hbox{\ensuremath|}_{#1}}

\newcommand{\tenAA}{\mathcal{A}}

\newcommand{\emphi}[1]{\emph{\textbf{#1}}}

\newcommand{\matAA}{\mathrm{A}}

\newcommand{\matMM}{\mathrm{M}}

\newcommand{\matBB}{\mathrm{B}}

\newcommand{\one}{\mathsf{1}}

\newcommand{\seqI}{\mathrm{I}}

\newcommand{\tuparg}[2]{{\Moplus}_{\ell}(#1^{\ell},#2^{\ell})}

\DeclareMathOperator*{\E}{\mathbf{E}}

\newcommand{\RR}{\mathbb{R}}      

\newcommand{\NN}{\mathbb{N}}      

\newcommand{\1}{\mathbf{1}}      

\newcommand{\defeq}{:=}
\newcommand{\R}{\RR}

\newcommand{\ftwo}[1]{\norm{2}{#1}}

\newcommand{\PExc}[2]{
\mathop{\widetilde{\mathbf{E}}_{#1}}{\!\pbrcx{#2}}}

\newcommand{\sos}[2]{\textsf{SoS}_{#1}\pth{#2}}
\newcommand{\hssos}[1]{\Lambda\pth{#1}}
\newcommand{\fmax}[1]{#1_{\max}}
\newcommand{\fmin}[1]{#1_{\min}}

\newcommand{\sfM}{\mathsf{M}}
\newcommand{\sfA}{\mathsf{A}}

\newcommand{\sfX}{\mathsf{X}}

\newcommand{\sfW}{\mathsf{W}}

\newcommand{\ha}{\widehat{A}}
\newcommand{\hA}{\widehat{\mathsf{A}}}

\newcommand{\hW}{\widehat{\mathsf{W}}}

\newcommand{\diag}{\mathsf{diag}}

\newcommand{\overbar}[1]{\mkern 1.2mu\overline{\mkern-1.2mu#1\mkern-1.2mu}\mkern 1.2mu}
\newcommand{\mi}[1]{\alpha({#1})}
\newcommand{\mindex}{{\NN}^{n}}
\newcommand{\degmindex}[1]{{\NN}_{#1}^{n}}
\newcommand{\udmindex}[1]{{\NN}_{\!\leq #1}^{n}}

\newcommand{\deffont}{\sf}
\newcommand{\defnt}[1]{{\deffont #1}}

\newcommand{\ie}{i.e.,\xspace}

\newcommand{\etal}{et al.\xspace}

\newcommand{\abs}[1]{\left\lvert #1 \right\rvert}

\begin{document}

\title{{\bf Sum-of-Squares Certificates for \\Maxima of Random Tensors on the Sphere}}


\author{
Vijay Bhattiprolu\thanks{Supported by NSF grants CCF-1422045 and CCF-1526092. \tt vpb@cs.cmu.edu} \and
Venkatesan Guruswami\thanks{Research supported in part by NSF grants CCF-1526092 and CCF-1563742. {\tt guruswami@cmu.edu} } \and
Euiwoong Lee\thanks{Supported by a Samsung Fellowship, a Simons Award for Graduate Students in Theoretical Computer Science, and NSF CCF-1526092. {\tt euiwoonl@cs.cmu.edu} }
}

\date{Computer Science Department \\ Carnegie Mellon University \\ Pittsburgh, PA 15213.}

\setcounter{page}{0}

\maketitle
\thispagestyle{empty}
\begin{abstract}
\noindent 
For an $n$-variate order-$d$ tensor $\tenAA$, define 
\[
\fmax{\tenAA} :=
\sup_{\| x \|_2 = 1} \langle \tenAA , x^{\otimes d} \rangle \ 
\]
to be the maximum value taken by the tensor on the unit sphere. It is known that for a random tensor with i.i.d $\pm 1$ entries,  $\fmax{\tenAA} \lesssim \sqrt{n\cdot d\cdot\log d}$ w.h.p.
We study the problem of efficiently certifying upper bounds on $\fmax{\tenAA}$ via the natural relaxation from the Sum of Squares (SoS) hierarchy. 
Our results include:
\begin{itemize}
\item When $\tenAA$ is a random order-$q$ tensor, we prove that 
$q$ levels of SoS certifies an upper bound $B$ on $\fmax{\tenAA}$ that satisfies
 \[ B ~~~~\leq~~ 
    \fmax{\tenAA}\cdot \pth{\frac{n}{q^{\,1-o(1)}}}^{q/4-1/2} \quad \text{w.h.p.} \]
Our upper bound improves a result of Montanari and Richard (NIPS 2014) when $q$ is large.
 \item We show the above bound is the best possible up to lower order terms, namely the optimum of the level-$q$ SoS relaxation is at least 
 \[ \fmax{\tenAA}\cdot \pth{\frac{n}{q^{\,1+o(1)}}}^{q/4-1/2} \ . \]
  
 \item When $\tenAA$ is a random order-$d$ tensor, we prove that 
$q$ levels of SoS certifies an upper bound $B$ on $\fmax{\tenAA}$ that satisfies 
\[
    B ~~\leq ~~ \fmax{\tenAA}\cdot \biggl(\frac{\widetilde{O}(n)}{q}\biggr)^{d/4 - 1/2}  \quad \text{w.h.p.}
\]
For growing $q$, this improves upon the bound 
certified by constant levels of SoS. This answers in part, a question posed by Hopkins, Shi, 
and Steurer (COLT 2015), who established the tight characterization for constant levels of SoS. 

\end{itemize}

\end{abstract}
\newpage

\section{Introduction}
It is a well-known fact from random matrix theory that 
for an $n \times n$ matrix $M$ whose entries are i.i.d Rademacher or
standard normal random variables, the maximum value $x^T M x$ taken by the associated quadratic form on the unit sphere $\|x\|_2=1$, is $\Theta(\sqrt{n})$ with high probability. Further, this maximum value can be computed efficiently for any matrix, as it equals the largest eigenvalue of $(M+M^T)/2$, so one can also efficiently certify that the maximum of a random quadratic form is at most $O(\sqrt{n})$.


This paper is motivated by the problem of analogous question for tensors. Namely, given a random order-$d$ tensor $\tenAA$ who entries are i.i.d random $\pm$ entries, we would like to certify an upper bound on the maximum value $\fmax{\tenAA} := \max_{\|x\|=1} \langle \tenAA, x^{\otimes d} \rangle$ taken by the tensor on the unit sphere.
This value is at most $O_d(\sqrt{n})$ with high probability~\cite{TS14}.
However, for $d \ge 3$, computing $\fmax{\tenAA}$ for a $d$-tensor $\tenAA$ is NP-hard, and it is likely that the problem is also very hard to approximate.  Assuming
the Exponential Time Hypothesis, Barak et al.~\cite{BBHKSZ12} proved
that computing $2 \rightarrow 4$ norm of a matrix, a special case of computing the norm of a $4$-tensor, is hard to approximate
within a factor $\exp(\log^{1/2-\epsilon}(n))$ for any $\epsilon > 0$.

Our goal is to certify an \emph{approximate} upper bound on $\fmax{\tenAA}$ is not too far from the true value. Specifically, we seek an estimate $B(\tenAA)$ which always upper bounds $\fmax{\tenAA}$, and with high probability is as close to $O_d(\sqrt{n})$ as possible for a random $\tenAA$.

In addition to its intrinsic interest, the problem of maximizing tensors and closely related tasks of computing tensor norms, has connections to diverse topics, such as 
 quantum information theory~\cite{BH13, BKS14}, the Small Set Expansion Hypothesis (SSEH) and the Unique Games Conjecture (UGC)~(via $2 \rightarrow 4$ norm, see \cite{BBHKSZ12, BKS14}), 
 refuting random CSPs~\cite{RRS16}, tensor decomposition~\cite{BKS15,GM15}, tensor PCA~\cite{MR14, HSS15}, and planted clique~(via the parity tensor, see \cite{FK08, BV09}). 
 Many of these applications are of considerable interest in the $2^{n^\epsilon}$-runtime regime. 
 


A natural approach to tackle the above problem is through the {\em
  Sum of Squares} (SoS) semidefinite programming relaxations.  There are several ways to represent a tensor $\tenAA \in \R^{[n]^d}$ (assume $d$ is even) in matrix form as $M \in \R^{[n]^{d/2} \times [n]^{d/2}}$  so that $\langle \tenAA, x^{\otimes d} \rangle = (x^{\otimes d/2})^T M x^{\otimes d/2}$ for all $x \in \R^n$. The largest eigenvalue $\lambda_{\max}(M)$ of any such matrix representation $M$ serves as an (efficiently computable) upper bound on $\fmax{\tenAA}$. The basic SoS relaxation looks for the best matrix representation, i.e., the one minimizing $\lambda_{\max}(M)$, among all possible representations of the tensor $\tenAA$. This can be expressed as a semidefinite program, and also has a natural dual view in terms of pseudo-expectations or moment matrices (see \ref{subsec:sos-prelims}).

The SoS hierarchy offers a sequence of relaxations,   
 parameterized by
the \emph{level} $q$, with larger $q$ giving a (potentially) tighter relaxation.  In our context, this amounts to optimizing over matrix representations of $\tenAA^{q/d}$ (we assume $q$ is divisible by $2d$); in the dual view, this involves optimizing over pseudo-expectations for polynomials of degree up to $q$ (as opposed to degree $d$ for the basic relaxation). The level-$q$ relaxation can be solved in $n^{O(q)}$ time by solving the associated semidefinite program. The SoS hierarchy thus presents a trade-off between approximation guarantee and runtime, with larger levels giving more accurate estimates at the expense of higher complexity. 

This work is concerned with both positive and negative results on the efficacy of the SoS hierarchy to approximately certify the maxima of random tensors. We now turn to stating our results formally.

\subsection{Our Results}
For an order-$q$ tensor $\tenAA\in (\Re^{n})^{\otimes d}$, 
the polynomial $\tenAA(x)$ and its maximum on the sphere $\fmax{\tenAA}$ are defined as
\[
\tenAA(x) := \iprod{\tenAA}{x^{\otimes d}}
\qquad 
\fmax{\tenAA} := \sup_{\|x\| = 1} \tenAA(x).
\]
When the entries of $\tenAA$ are i.i.d Rademacher random variables (or i.i.d. Gaussians), 
it is known that $\fmax{\tenAA} \lesssim \sqrt{n\cdot d\cdot \log d}$ (see~\cite{TS14}).
We will also use, for a polynomial $g$, $\fmax{g}$ to denote $\sup_{\|x\| = 1} g(x)$.


\medskip \noindent \textbf{SoS degree = Polynomial Degree.}

\smallskip\noindent We study the performance of degree-$q$ SoS on random tensors of order-$q$. The formal definition and basic properties of SoS relaxations are presented in Section~\ref{subsec:sos-prelims}.

\begin{theorem}
\thmlab{q-tensor:informal}
For any even $q\leq n$, 
let $\tenAA\in (\Re^{n})^{\otimes q}$ be a $q$-tensor with independent, Rademacher entries. 
With high probability, the value $B$ of the degree-$q$ SoS relaxation of $\fmax{\tenAA}$ satisfies
\[
  2^{-O(q)} \cdot \pth{\frac{n}{q}}^{q/4-1/2} 
    \leq~~~~ \frac{B}{\fmax{\tenAA}} ~~~~\leq~~ 
    2^{O(q)} \cdot \pth{\frac{n}{q}}^{q/4-1/2}.
\]
\end{theorem}
This improves upon the $O(n^{q/4})$ upper bound by Montanari and Richard~\cite{MR14}.

\medskip \noindent \textbf{SoS Degree $\gg$ Polynomial Degree.}
\begin{theorem}
\thmlab{d-tensor:informal}
Let $\tenAA\in (\Re^{n})^{\otimes d}$ be a $d$-tensor with independent, Rademacher 
entries. Then for any even $q$ satisfying $d\leq q\leq n$, with high probability, the 
degree-$q$ SoS certifies an upper bound $B$ on $\fmax{\tenAA}$ where w.h.p., 
\[
    \frac{B}{\fmax{\tenAA}} ~~\leq ~~ \pth{\frac{\widetilde{O}(n)}{q}}^{d/4-1/2}
\]
\end{theorem}

\begin{remark}
Combining our upper bounds with the work of \cite{HSS15} would yield improved 
tensor-PCA guarantees on higher levels of SoS. Our techniques prove similar results for a more general random model where each coefficient is 
independently sampled from a centred subgaussian distribution. See the previous version of the 
paper~\cite{BGL16} for details.
\end{remark}

\begin{remark}
Raghavendra, Rao, and Schramm~\cite{RRS16} have independently and concurrently obtained similar
(but weaker) results to \thmref{d-tensor:informal} for random degree-$d$ polynomials. 
Specifically, their upper bounds appear to require the assumption that the SoS level $q$ must be 
less than $n^{1/(3d^2)}$ (our result only assumes $q\leq n$). Further, they certify an 
upper bound that matches \thmref{d-tensor:informal} only when $q \leq 2^{\sqrt{\log n}}$.  
\end{remark}

\subsection{Related Work}
\noindent \textbf{Upper Bounds.}
Montanari and Richard~\cite{MR14} presented a $n^{O(d)}$-time algorithm that can certify that the
optimal value of $\fmax{\tenAA}$ for a random $d$-tensor is at most $O(n^{ \frac{ \lceil d/2 \rceil}{2}})$ with high probability. Hopkins, Shi,
and Steurer~\cite{HSS15} improved it to $O(n^{ \frac{ d }{4}})$ with the same running time. They
also asked how many levels of SoS are required to certify a bound of $n^{3/4 - \delta}$ for $d = 3$. 

Our analysis asymptotically improves the aforementioned bound when $q$ is growing with $n$, and 
we prove an essentially matching lower bound (but only for the case $q=d$).  
Secondly, we consider the case when $d$ is fixed, and give
improved results for the performance of degree-$q$ SoS (for large $q$), thus answering in part, a
question posed by Hopkins, Shi and Steurer~\cite{HSS15}. 

Raghavendra, Rao, and Schramm~\cite{RRS16} also prove results analogous to 
 \thmref{d-tensor:informal} for the case of \emph{sparse} random polynomials 
(a model we do not consider in this work, and which appears to pose additional technical difficulties). This implied upper bounds for refuting random instances of constraint satisfaction problems using higher levels of the SoS hierarchy, which were shown to be tight via matching  SoS lower bounds in \cite{KMOW17}.

\medskip \noindent \textbf{Lower Bounds.}
While we only give lower bounds for the case of $q=d$, subsequent to our work, 
Hopkins \etal ~\cite{HKP16} proved the following theorem, which gives lower 
bounds for the case of $q\gg d$:
\begin{theorem}
Let $f$ be a degree-$d$ polynomial with i.i.d. gaussian coefficients. If there is  
some constant $\epsilon > 0$ such that $q\geq n^{\epsilon}$, then with high probability over 
$f$, the optimum of the level-$q$ SoS relaxation of $\fmax{f}$ is at least
\[
\fmax{f}\cdot 
\Omega_d \pth{(n/q^{O(1)})^{d/4 - 1/2}} \ .
\]
\end{theorem}

Note that this almost matches our upper bounds from \thmref{d-tensor:informal}, modulo the 
exponent of $q$. For this same reason, the above result does not completely recover our lower 
bound in \thmref{q-tensor:informal} for the special case of $q=d$. 

\paragraph*{Results for worst-case tensors.} It is proved in 
\cite{BGGLT16} that the $q$-level SoS gives an $(O(n)/q)^{d/2-1}$ approximation to 
$\ftwo{\tenAA}$ in the case of arbitrary $d$-tensors and an $(O(n)/q)^{d/4-1/2}$ approximation to $\fmax{\tenAA}$
in the case of $d$-tensors with non-negative entries (for technical reasons one can only 
approximate $\ftwo{\tenAA}=\max\{|\fmax{\tenAA}|,|\fmin{\tenAA}|\}$ in the former case). 

It is interesting to note that the approximation factor in the case of non-negative tensors 
matches the approximation factor (upto polylogs) we achieve in the random case. Additionally, 
the gap given by \thmref{q-tensor:informal} for the case of random tensors provides the best 
degree-$q$ SoS gap for the problem of approximating the $2$-norm of arbitrary $q$-tensors. 
Hardness results for the arbitrary tensor $2$-norm problem is an important pursuit due to 
its connection to various problems for which subexponential algorithms are of interest.

\subsection{Organization}
We begin by setting some important notation concerning SoS matrices, and describe some basic preliminaries about the SoS hierarchy in Section~\ref{sec:notation-prelims}. We touch upon the main technical ingredients driving our work, and give an overview of the proof of \thmref{d-tensor:informal} and the lower bound in \thmref{q-tensor:informal} in Section~\ref{sec:overview}. We present the proof of 
\thmref{d-tensor:informal} for the case of even $d$ in Section~\ref{sec:ub-even}, with the more tricky odd $d$ case handled in Appendix~\ref{app:3ten:ub}. The lower bound on the value of SoS-hierarchy claimed in \thmref{q-tensor:informal} is proved in \secref{qten:lb}, and the upper bound in \thmref{q-tensor:informal} also follows based on some techniques in that section.

\section{Notation and Preliminaries}
\label{sec:notation-prelims}
\noindent \textbf{Multi-index and Multiset.}
A \defnt{multi-index} is defined as a sequence $\alpha \in \mindex$. We use
$\abs{\alpha}$ to denote $\sum_{i=1}^n \alpha_i$ and $\degmindex{d}$ (resp. $\udmindex{d}$) to
denote the set of all multi-indices $\alpha$ with $\abs{\alpha} = d$ (resp. $\abs{\alpha} \leq d$). 
We use $\1$ to denote the multi-index $1^n$.
Thus, a homogeneous polynomial $f$ of degree $d$ can be expressed in terms of its coefficients as

\smallskip
$\qquad\qquad
f(x) ~=~ \sum_{\alpha \in \degmindex{d}} f_{\alpha} \cdot x^{\alpha} ,$

\smallskip\noindent
where $x^{\alpha}$ is used to denote the monomial corresponding to $\alpha$. 
In general, with the exception of absolute-value, 
any scalar function/operation when applied to vectors/multi-indices, returns the vector obtained by applying the function/operation entry-wise.

\subsection{Matrices}
For $k \in \NN$, we will consider $[n]^k \times [n]^k$ matrices $M$ with real entries. All matrices 
considered in this paper should be taken to be symmetric (unless otherwise stated). We index entries 
of the matrix $M$ as $M[I,J]$ by \defnt{tuples} $I,J \in [n]^k$. $\oplus$ denotes tuple-concatenation. 

A tuple $I = (i_1, \ldots, i_k)$ naturally corresponds to a multi-index $\mi{I} \in \degmindex{k}$
with $\abs{\mi{I}} = k$, i.e. $\mi{I}_j = |\{\ell ~|~i_{\ell} = j\}|$.  For a tuple $I \in [n]^k$, 
we define $\orbit{I}$ the set of all tuples $J$ which correspond to the same multi-index \ie 
$\alpha(I) = \alpha(J)$. Thus, any multi-index $\alpha\in \degmindex{k}$,  corresponds to an 
equivalence class in $[n]^k$. We also use $\orbit{\alpha}$ to denote the class of all tuples 
corresponding to $\alpha$.

Note that a matrix of the form $\pth{x^{\otimes k}} \pth{x^{\otimes k}}^T$ has many additional
symmetries, which are also present in solutions to programs given by the SoS hierarchy.   
To capture this, consider the following definition: 

\begin{defn}[SoS-Symmetry]
A matrix $\sfM$ which satisfies $\sfM[I,J] = \sfM[K,L]$ whenever 
$\alpha(I) + \alpha(J) = \alpha(K) + \alpha(L)$ is referred to as \defnt{SoS-symmetric}. 
\end{defn}

\begin{defn}[Matrix-Representation]
For a homogeneous degree-$t$ ($t$ even) polynomial $g$, we say a matrix $\matMM_g\in \Re^{[n]^{t/2}\times [n]^{t/2}}$ 
is a degree-$t$ matrix representation of $g$ if for all $x$, $g(x) = (x^{\otimes t/2})^{T}~\matMM_g~x^{\otimes t/2}$. 
(We note here that every homogeneous polynomial has a unique SoS-Symmetric matrix representation.)
\end{defn}
Note that $\lambda_{\max}(\matMM_g)$ is an upper bound on $\fmax{g}$. 
This prompts the following relaxation of $\fmax{g}$ that is closely related to the final SoS relaxation used 
in our upper bounds:

\begin{defn}
For a homogeneous degree-$t$ ($t$ even) polynomial $g$, define 
\[
    \hssos{g} ~\defeq~ \inf \big\{\lambda_{\max}(M_g) \sep{M_g \text{ represents } g}\big\} 
\]
\end{defn}
As we will see shortly, $\hssos{g}$ is the dual of a natural SoS relaxation of $\fmax{g}$.

\subsection{SoS Hierarchy}
\label{subsec:sos-prelims}
Let $\Re[x]_{\leq q}$ be the vector space of polynomials with real coefficients in variables 
$x = (x_1, \dots, x_n)$, of degree at most $q$. For an even integer $q$, the degree-$q$ 
pseudo-expectation operator is a linear operator $\PE : \Re[x]_{\leq q} \mapsto \Re$ such that 
\smallskip

\begin{compactenum}[1.]
\item $\PEx{1} = 1$ for the constant polynomial $1$.
\item $\PEx{p_1 + p_2} = \PEx{p_1} + \PEx{p_2}$ for any polynomials $p_1, p_2 \in \Re[x]_{\leq q}$. 
\item $\PEx{p^2} \geq 0$ for any polynomial $p \in \Re[x]_{\leq q/2}$. 
\end{compactenum}

The pseudo-expectation operator $\PE$ can be completely described by the \emphi{moment matrix} (while $x$ is a column vector, we abuse notation and let $(1,x)$ denote the column vector $(1,x_1,\dots ,x_n)^T$)
\begin{equation}
\label{eq:mom-matrix-inhom}
\overbar{\sfX}:= \PEx{(1,x)^{\otimes q/2} \,\,((1,x)^{\otimes q/2})^T} \  .
\end{equation}
Moreover, the condition $\PEx{p^2}\geq 0$ for all $p \in \Re[x]_{\leq q/2}$ can be shown to be 
equivalent to $\overbar{\sfX}\succeq 0$. 

\paragraph{Constrained Pseudoexpectations.}
For a system of polynomial constraints \[C = \brc{f_1 = 0, \ldots, f_m = 0, g_1 \geq 0, \ldots,
  g_r \geq 0},\] we say $\PE_C$ is a pseudoexpectation operator respecting $C$, if in addition to the
above conditions, it also satisfies
\begin{enumerate}
\item $\PExc{C}{p \cdot f_i} = 0$,  $\forall i \in [m]$ and $\forall p$ such that $\deg(p \cdot f_i)
  \leq q$.
\item $\PExc{C}{p^2 \cdot \prod_{i \in S} g_i} \geq 0$, $\forall S \subseteq [r]$ and $\forall p$
  such that $\deg(p^2 \cdot \prod_{i \in S} g_i) \leq q$.
\end{enumerate}
It is well-known that such constrained pseudoexpectation operators can be described as solutions to
semidefinite programs of size $n^{O(q)}$ \cite{BS14, Laurent09}. This hierarchy of semidefinite
programs for increasing $q$ is known as the SoS hierarchy.

\paragraph{Additional Facts about SoS. }
We shall record here some well-known facts about SoS that come in handy later. 

\begin{claim}
For polynomials $p_1,p_2$, let $p_1\succeq p_2$ denote that 
$p_1-p_2$ is a sum of squares. 
It is easy to verify that if $p_1,p_2$ are homogeneous degree $d$ polynomials and 
there exist matrix representations $M_{p_1}$ and $M_{p_2}$ of $p_1$ and $p_2$ respectively, such that ~
$M_{p_1}-M_{p_2}\succeq 0$, \,then ~$p_1-p_2\succeq 0$. 
\end{claim}

\begin{claim}[Pseudo-Cauchy-Schwarz \cite{BKS14}]
$\PEx{p_1p_2}\leq (\PEx{p_1^2}\PEx{p_2^2})^{1/2}$ for any $p_1,p_2$ of degree at most $q/2$. 
\end{claim}

\paragraph{SoS Relaxations for $\fmax{\tenAA}$.}~\\
Given an order-$q$ tensor $\tenAA$, our degree-$q$ SoS relaxation for $\fmax{\tenAA}$ which we will 
henceforth denote by $\sos{q}{\tenAA(x)}$ is given by, 
\begin{align*}
    \mathsf{maximize} \qquad \qquad \qquad \quad \PExc{C}{\tenAA(x)}  \\
    \mathsf{subject ~to:}  \qquad \qquad \widetilde{\mathbf{E}}_C ~\text{is a degree-$q$} \\
    \text{pseudoexpectation} \\
    \widetilde{\mathbf{E}}_C ~\text{respects}~ C \equiv \brc{\norm{2}{x}^q = 1}
\end{align*}

\noindent
Assuming $q$ is divisible by $2d$, we make an observation that is useful in our upper bounds: 
\begin{equation}
\Eqlab{relaxations:relations}
    \fmax{\tenAA}
    ~~\leq ~~
    \sos{q}{\tenAA(x)} 
    ~~\leq ~~
    \sos{q}{\tenAA(x)^{\,q/d}}^{d/q}
    =~~
    \hssos{\tenAA(x)^{\,q/d}}^{d/q}
\end{equation}
where the second inequality follows from Pseudo-Cauchy-Scwarz, and the equality follows from 
well known strong duality of the following programs (specifically, take 
$g(x) := \tenAA(x)^{\,q/d}$)~:\footnote{Compared to \eqref{eq:mom-matrix-inhom}, the primal formulation here uses a \emph{homogeneous} moment matrix or pseudo-expectation operator, defined for polynomials of degree exactly $q$.}
\begin{figure}[htb]
\begin{tabular}{|c|}
\hline
\begin{minipage}[t]{0.94\textwidth}
\smallskip
\underline{\textsf{Dual}}
\[
\hssos{g} ~\defeq~ \inf \big\{\lambda_{\max}(M_g) \sep{M_g \text{ represents } g}\big\} 
\]
\smallskip
\end{minipage} \\
\hline
\begin{minipage}{0.92\textwidth}
\begin{tabular}{l|l}
\begin{minipage}[t]{0.48 \textwidth}
\smallskip
\underline{\textsf{Primal I}}
\begin{align*}
\mathsf{maximize} ~~~\qquad \qquad \qquad \iprod{\sfM_g}{\sfX}  \\
\mathsf{subject ~to:} \qquad \qquad \quad \Tr{\sfX} = 1 \\
{\sfX} ~\text{is SoS symmetric} \\
\sfX \succeq 0
\end{align*}
\end{minipage}
&
\begin{minipage}[t]{0.48 \textwidth}
\smallskip
\underline{\textsf{Primal II}}
\begin{align*}
\mathsf{maximize} \qquad \qquad \qquad \qquad \quad \PExc{C}{g}  \\
\mathsf{subject ~to:}  \qquad \qquad \widetilde{\mathbf{E}}_C ~\text{is a degree-$q$} \\
 \text{pseudoexpectation} \\
 \widetilde{\mathbf{E}}_C ~\text{respects}~ C \equiv \brc{\norm{2}{x}^q = 1}\\
\end{align*}
\end{minipage}  \\
\end{tabular}
\end{minipage} \\
\hline
\end{tabular}
\caption{Duals of $\hssos{g}$ for the degree-$q$ homogeneous polynomial $g$}
\label{fig:Lambda}
\end{figure}

\smallskip \noindent \textbf{Note.}
In the rest of the paper, we will drop the subscript $C$ of the pseudo-expectation 
operator since throughout this work, we only assume the hypersphere constraint.

\section{Overview of our Methods}
\label{sec:overview}
We now give a high level view of the two broad techniques driving this work, followed by a more detailed overview of the proofs.

\smallskip
\noindent\textbf{Higher Order Mass-Shifting.} Our approach to upper bounds on a random low degree (say $d$) polynomial $f$, is through exhibiting a matrix representation of $f^{q/d}$ that has 
small operator norm. Such approaches had been used previously for low-degree SoS upper bounds. However when the SoS degree is constant, the set of SoS symmetric positions is also a constant 
and the usual approach is to shift all the mass towards the diagonal which is of little consequence when the SoS-degree is low. In contrast, when the SoS-degree is large, many non-trivial issues arise when shifting mass across SoS-symmetric positions, as there are many permutations with very large operator norm. In our setting, mass-shifting approaches like symmetrizing and diagonal-shifting fail quite spectacularly to provide good upper bounds. For our upper bounds, we crucially exploit the existence of "good permutations", and moreover that there are $q^{q}\cdot 2^{-O(q)}$ such good permutations. On averaging the representations corresponding to these good permutations, we obtain a 
matrix that admits similar spectral preperties to those of a matrix with i.i.d. entries, and with much lower variance (in most of the entries) compared to the naive representations. 

\medskip\noindent\textbf{Square Moments of Wigner Semicircle Distribution.}
Often when one is giving SoS lower bounds, one has a linear functional that is not necessarily PSD and a natural approach is to fix it by adding a pseudo-expectation operator with large value on square polynomials (under some normalization). Finding such operators however, is quite a non-trivial task when the SoS-degree is growing. We show that if $x_1, \dots, x_n$ are independently drawn from the Wigner semicircle distribution, then for any polynomial $p$ of any degree, $\Ex{p^2}$ is large (with respect to the degree and coefficients of $p$).
Our proof crucially relies on knowledge of the Cholesky decomposition of the moment matrix of the univariate Wigner distribution. This tool was useful to us in giving tight $q$-tensor lower bounds, and we believe it to be generally useful for high degree SoS lower bounds.

\subsection{Overview of Upper Bound Proofs}
\seclab{ub:overview}
For even $d$, let $\tenAA \in \Re^{[n]^{d}}$ be a $d$-tensor with i.i.d. $\pm 1$ entries and let $A \in \Re^{[n]^{d/2} \times [n]^{d/2}}$ be the matrix flattening of $\tenAA$, \ie 
$A[I,J] = \tenA{I\oplus J}$ (recall that $\oplus$ denotes tuple concatenation). Also let $f(x) := \tenAA(x) = \iprod{\tenAA}{x^{\otimes d}}$. 
It is well known that $\fmax{f} \leq O(\sqrt{n\cdot d\cdot \log d})$ with high probability~\cite{TS14}.
For such a polynomial $f$ and any $q$ divisible by $d$, in order to establish \thmref{d-tensor:informal}, by \Eqref{relaxations:relations} it is sufficient to prove that with high probability, 
\[
\pth{\hssos{f^{q/d}}}^{d/q} ~\leq~ \widetilde{O}\pth{\frac{n}{q^{1 - 2/d}}}^{d/4} = ~~\widetilde{O}\pth{\frac{n}{q}}^{d/4 - 1/2} \cdot \fmax{f}.
\] 

We give an overview of the proof. Let $d = 4$ for the sake of clarity of exposition. 
To prove an upper bound on $\hssos{f^{q/4}}$ using degree-$q$ SoS (assume $q$ is a multiple of $4$), we define a suitable matrix representation $M := M_{f^{q/4}} \in \RR^{[n]^{q/2} \times [n]^{q/2}}$ of $f^{q/4}$ and bound $\norm{2}{M}$. Since $\hssos{f} \leq (\norm{2}{M})^{q/4}$ for any representation $M$, a good upper bound on $\norm{2}{M}$ certifies that $\hssos{f}$ is small. 

One of the intuitive reasons taking a high power gives a better bound on the spectral norm is that this creates more entries of the matrix that correspond to the same monomial, and distributing the coefficient of this monomial equally among the corresponding entries reduces variance (i.e., $\Var{X}$ is less than $k \cdot \Var{X / k}$ for $k > 1$). 
In this regard, the most natural representation $M$ of $f^{q/4}$ is the {\em complete symmetrization}. 
\begin{align*}
&\quad M_{c}[(i_1, \dots, i_{q/2}), (i_{q/2+1}, \dots, i_q)] \\
&= 
\frac{1}{q!} \cdot \sum_{\pi \in \Sym_{q}} 
A^{\otimes q/4} [(i_{\pi(1)}, \dots, i_{\pi(q/2)}), (i_{\pi(q/2+1)}, \dots, i_{\pi(q)})] \\
&= 
\frac{1}{q!} \cdot \sum_{\pi \in \Sym_{q}} \prod_{j = 1}^{q / 4} A[(i_{\pi(2j - 1)}, i_{\pi(2j)}), (i_{\pi(q/2 + 2j - 1)}, i_{\pi(q/2 + 2j)})].
\end{align*}

However, $\norm{2}{M_{c}}$ turns out to be much larger than $\hssos{f}$, even when $q = 8$. One intuitive explanation is that 
$M_{c}$, as a $n^4 \times n^4$ matrix, contains a copy of $\Vector{A} \Vector{A}^T$, 
where $\Vector{A} \in \RR^{[n]^4}$ is the vector with $\Vector{A}[i_1, i_2, i_3, i_4] = A[(i_1, i_2), (i_3, i_4)]$. 
Then $\Vector{A}$ is a vector that witnesses $\norm{2}{M_{c}} \geq \Omega(n^2)$, regardless of the randomness of $f$. 
Our final representation 
\footnote{the independent and concurrent work of \cite{RRS16} uses the 
same representation}
is the following {\em row-column independent symmetrization} that simultaneously respects the spectral structure of a random matrix $A$ and reduces the variance. Our $M$ is given by 
\begin{align*}
&\quad M[(i_1, \dots, i_{q/2}), (j_{1}, \dots, j_{q/2})] \\
&= 
\frac{1}{(q/2)!^2} \cdot \sum_{\pi, \sigma \in \Sym_{q/2}} 
A^{\otimes q/4} [(i_{\pi(1)}, \dots, i_{\pi(q/2)}), (j_{\sigma(1)}, \dots, j_{\sigma(q/2)})] \\
&= \frac{1}{(q/2)!^2} \cdot \sum_{\pi, \sigma \in \Sym_{q/2}} \prod_{k = 1}^{q / 4} A[(i_{\pi(2k - 1)}, i_{\pi(2k)}), (j_{\sigma(2k - 1)}, j_{\sigma(2k)})].
\end{align*}

To formally show $\norm{2}{M} = \tilde{O}(n / \sqrt{q})^{q/4}$ with high probability, we use the trace method to show
\[
\Ex{\Tr{M^{p}}} \leq 2^{O(pq \log p)} \frac{n^{pq / 4 + q/2}}{q^{pq / 8}},
\] 
where $\Ex{\Tr{M^{p}}}$ can be written as (let $I^{p + 1} := I^1$)
\begin{align*}
&\quad 
\Ex{
\sum_{I^1, \dots, I^p \in [n]^{q/2}} \prod_{j=1}^p M[I^j, I^{j+1}] 
} \\ 
&= 
\sum_{I^1, \dots, I^p} \Ex{ \prod_{j=1}^p 
(\sum_{\pi_j, \sigma_j \in \Sym_{q/2}} \prod_{k = 1}^{q/4} A[(I^k_{\pi_j(2k - 1)}, I^k_{\pi_j(2k)}), (I^{k + 1}_{\sigma_j(2k - 1)}, I^{k + 1}_{\sigma_j(2k)})] )  
}.
\end{align*}
Let $E(I^1, \dots, I^p)$ be the expectation value for $I^1, \dots, I^p$ in the right hand side. 
We study $E(I^1, \dots, I^p)$ for each $I^1, \dots, I^p$ by careful counting of the number of permutations on a given sequence with possibly repeated entries. 
For any $I^1,\dots ,I^p\in [n]^{q/2}$, ~let $\numdist{I^{1},\dots ,I^{p}}$ denote the 
number of distinct elements of $[n]$ that occur in $I^1,\dots ,I^p$, and for each $s = 1, \dots, \numdist{I^{1},\dots ,I^{p}}$, let $c^s \in (\{ 0 \} \cup [q/2])^p$ denote the number of times that the $j$th smallest element occurs in $I^1, \dots, I^p$. 
When $E(I^1, \dots, I^p) \neq 0$, it means that for some permutations $\{ \pi_j, \sigma_j \}_{j}$, every term $A[\cdot, \cdot]$ must appear even number of times. This implies that the number of distinct elements in $I^1, \dots, I^p$ is at most half the maximal possible number $pq/2$. This lemma proves the intuition via graph theoretic arguments. 

\begin{lemma}
If $E(I^1, \dots, I^p) \neq 0$,  
$
\numdist{I^{1},\dots ,I^{p}} \leq \frac{pq}{4} + \frac{q}{2}.
$
\lemlab{random_overview_graph}
\end{lemma}
The number of $I^1, \dots, I^p$ that corresponds to a sequence $c^1, \dots, c^s$ is at most 
$
\frac{n^s}{s!} \cdot \frac{((q/2)!)^p}{\prod_{\ell \in [p]} c^1_{\ell}! \cdot c^p_{\ell}!}
$.
Furthermore, there are at most 
$2^{O(pq)} p^{pq/2}$
different choices of $c^1, \dots, c^s$ that corresponds to some $I^1, \dots, I^p$. The following technical lemma bounds $E(I^1, \dots, I^p)$ by careful counting arguments. 
\begin{lemma}
For any $I^1, \dots, I^p$, 
$E(I^1, \dots, I^p) \leq 2^{O(pq)} \frac{p^{5pq / 8}}{q^{3pq / 8}} \prod_{\ell \in [p]} c^1_{\ell}! \dots c^s_{\ell}!$.
\lemlab{random_overview_permutation}
\end{lemma}
Summing over all $s$ and multiplying all possibilities, 
\begin{align*}
\Ex{\Tr{M^{p}}}
&\leq \sum_{s = 1}^{pq/4 + q/2} \bigg(
2^{O(pq)} p^{pq/2} \bigg) \cdot \bigg(
\frac{n^s}{s!} \cdot ((q/2)!)^p \bigg) \cdot 
\bigg( 2^{O(pq)} \frac{p^{5pq / 8}}{q^{3pq / 8}} \bigg) \\
&= ~\max_{1\leq s\leq pq/4 + q/2} 
~2^{O(pq \log p)} \cdot n^s \cdot \frac{q^{pq/8}}{s!}.
\end{align*}
When $q \leq n$, the maximum occurs when $s = pq / 4 + q / 2$, so $\Ex{\Tr{M^p}} \leq 2^{O(pq \log p)} \cdot \frac{n^{pq/4 + q/2}}{q^{pq/8}}$ as desired.

\subsection{Overview of Lower Bound Proofs}
Let $\tenAA,A,f$ be as in \secref{ub:overview}. 
To prove the lower bound in~\thmref{q-tensor:informal}, we construct a moment matrix $\sfM$ that is positive semidefinite, SoS-symmetric, $\Tr{\sfM} = 1$, and $\langle A, \sfM \rangle \geq 2^{-O(d)} \cdot \frac{n^{d/4}}{d^{d/4}}$. 
At a high level, our construction is $\sfM := c_1 \sfA + c_2 \sfW$ for some $c_1, c_2$, where $\sfA$ contains entries of $A$ only corresponding to the multilinear indices, averaged over all SoS-symmetric positions. This gives a large inner product with $A$, SoS-symmetry, and nice spectral properties even though it is not positive semidefinite. 
The most natural way to make it positive semidefinite is adding a copy of the identity matrix, but this will again break the SoS-symmetry.

Our main technical contribution here is the construction of $\sfW$ that acts like a {\em SoS-symmetrized identity}. It has the minimum eigenvalue at least $\frac{1}{2}$, while the trace being $n^{d/2} \cdot 2^{O(d)}$, so the ratio of the average eigenvalue to the minimum eigenvalue is bounded above by $2^{O(d)}$, which allows us to prove a tight lower bound. 
To the best of our knowledge, no such bound was known for SoS-symmetric matrices except small values of $d = 3, 4$. 

Given $I, J \in [n]^{d/2}$, we let $\sfW[I, J] := \E[x^{\mi{I} + \mi{J}}]$, where $x_1, \dots, x_n$ are independently sampled from the {\em Wigner semicircle distribution}, whose probability density function is the semicircle $f(x) = \frac{2}{\pi}\sqrt{1 - x^2}$. 
Since $\E[x_1^{\ell}] = 0$ if $\ell$ is odd and $\E[x_1^{2\ell}] = \frac{1}{\ell + 1}\binom{2\ell}{\ell}$, which is the $\ell$th Catalan number, each entry of $\sfW$ is bounded by $2^{O(d)}$ and $\Tr{\sfW} \leq n^{d/2} \cdot 2^{O(d)}$. 
To prove a lower bound on the minimum eigenvalue, we show that for any degree-$\ell$ polynomial $p$ with $m$ variables, $\E[p(x_1, \dots, x_m)^2]$ is large by induction on $\ell$ and $m$. 
We use another property of the Wigner semicircle distribution that if $H \in \RR^{(d + 1) \times (d + 1)}$ is the univariate moment matrix of $x_1$ defined by $H[i, j] = \E[x_1^{i + j}]$ ($0 \leq i, j \leq d$) and $H = (R^T)R$ is the Cholesky decomposition of $H$, $R$ is an upper triangular matrix with $1$'s on the main diagonal. This nice Cholesky decomposition allows us to perform the induction on the number of variables while the guarantee on the minimum eigenvalue is independent of $n$. 


\section{Upper bounds for even degree tensors}
\label{sec:ub-even}
For even $d$, let $\tenAA \in \Re^{[n]^{d}}$ be a $d$-tensor with i.i.d. $\pm 1$ entries and let $A \in \Re^{[n]^{d/2} \times [n]^{d/2}}$ be the matrix flattening of $\tenAA$, \ie 
$A[I,J] = \tenA{I\oplus J}$ (recall that $\oplus$ denotes tuple concatenation). Also let $f(x) := \tenAA(x) = \iprod{\tenAA}{x^{\otimes d}}$. 
With high probability $\fmax{f} = O(\sqrt{n\cdot d\cdot \log d})$. In this section, we prove that for every $q$ divisible by $d$, with high probability,  
\[
\pth{\hssos{f^{q/d}}}^{d/q} ~\leq~ \widetilde{O}\pth{\frac{n}{q^{1 - 2/d}}}^{d/4} =~~ \widetilde{O}\pth{\frac{n}{q}}^{d/4 - 1/2} \cdot \fmax{f}.
\] 
To prove it, we use the following matrix representation $M$ of $f^{q/d}$, and show that $\norm{2}{M} \leq \tilde{O}_d\pth{\pth{\frac{n \log^{5} n}{q^{1 - 2/d}}}^{q/4}}$.
Given a tuple $I = (i_1, \dots, i_q)$, and an integer $d$ that divides $q$ and $1 \leq \ell \leq q/d$, let $I_{\ell ; d}$ be the $d$-tuple $(I_{d(\ell - 1) + 1}, \dots, I_{d \ell})$ (i.e., if we divide $I$ into $q/d$ tuples of length $d$, $I_{\ell ; d}$ be the $\ell$-th tuple). 
Furthermore, given a tuple $I = (i_1, \dots, i_q) \in [n]^q$ and a permutation $\pi \in [n]^q$, let $\pi (I)$ be another $q$-tuple whose $\ell$th coordinate is $\pi(i_{\ell})$. 
For $I, J \in [n]^{q/2}$, $M[I, J]$ is formally given by  
\begin{align*}
M[I, J] &= 
\frac{1}{q!} \cdot \sum_{\pi, \sigma \in \Sym_{q/2}} 
A^{\otimes q/d} [\pi ( I), \sigma ( J)] \\
&= \frac{1}{q!} \cdot \sum_{\pi, \sigma \in \Sym_{q/2}} \prod_{\ell = 1}^{q / d} A[(\pi ( I))_{\ell ; d/2} , (\sigma ( J))_{\ell ; d/2}].
\end{align*}
We perform the trace method to bound $\norm{2}{M}$. Let $p$ be an even integer, that will be eventually taken as $\Theta(\log n)$. $\Tr{M}$ can be written as (let $I^{p + 1} := I^1$)
\begin{align*}
&\quad 
\Ex{
\sum_{I^1, \dots, I^p \in [n]^{q/2}} \prod_{\ell =1}^p M[I^{\ell}, I^{\ell+1}] 
} \\
&= 
\sum_{I^1, \dots, I^p} \Ex{ \prod_{\ell =1}^p 
(\sum_{\pi_j, \sigma_j \in \Sym_{q/2}} \prod_{m = 1}^{q/d} A[(\pi ( I^{\ell}))_{m; d/2}, (\sigma ( I^{\ell + 1}))_{m; d/2})] )  
}.
\end{align*}

Let $E(I^1, \dots, I^p) := \Ex{\prod_{\ell=1}^p M[I^{\ell}, I^{\ell+1}] }$, which is the expected value in the right hand side. 
To analyze $E(I^1, \dots, I^p)$, we first introduce notions to classify $I^1, \dots, I^p$ depending on their intersection patterns. 
For any $I^1,\dots ,I^p\in [n]^{q/2}$, let $e_k$ denote the $k$-th smallest 
element in $\bigcup\limits_{\ell,\,j} \{i^{\ell}_j\}$. For any $c^1,\dots ,c^s\in [q/2]^p$, let 
\begin{align*}
    &\mathcal{C}(c^1\dots c^s) := \\
    &\quad 
    \brc{(I^1,\dots ,I^p)\sep{\numdist{I^1,\dots ,I^p}=s,
    ~\forall k\in [s], 
    \ell\in [p], ~\text{$e_k$ appears $c^k_{\ell}$ times in 
    $I^{\ell}$}}}.
\end{align*}
The following two observations on $c^1, \dots, c^s$ can be easily proved. 

\begin{observation}
\obslab{num:guesses}
If $\mathcal{C}(c^1,\dots ,c^s)\neq \phi$, 
\[
        \cardin{\mathcal{C}(c^1,\dots ,c^s)} 
        \leq 
        \frac{n^{s}}{s!}\times 
        \frac{((q/2)!)^{p}}{\prod\limits_{\ell\in[p]} 
        c_\ell^{1}!\dots c_\ell^{s}!}.
\] 
Moreover, 
\[
        \cardin{\brc{(c^1,\dots ,c^s)\in ([q/2]^p)^{s}
        \sep{\mathcal{C}(c^1,\dots ,c^s)\neq \phi}}} 
        \leq 
        2^{O(pq)} p^{\,pq/2}.
\]
\end{observation}
The following lemma bounds $E(I^1, \dots, I^p)$ in terms of the corresponding $c_1, \dots, c_s$. 

\begin{lemma}
\lemlab{fixed:mult}
    Consider any $c^1,\dots ,c^s\in [q/2]^p$ and 
    $(I^{1},\dots ,I^{p})\in \mathcal{C}(c^1,\dots ,c^s)$. We have 
    \begin{align*}
        E(I^1, \dots, I^p)
        &\leq 
        2^{O(pq)}\,\frac{p^{1/2 + 1/2d}}{q^{1/2 - 1/2d}}\,
        \prod_{\ell\in [p]} c_\ell^{1}!\dots c_\ell^{s}!
    \end{align*}
\end{lemma}

\begin{proof}
    Consider any $c^1,\dots ,c^s\in [q/2]^p$ and 
    $(I^{1},\dots ,I^{p})\in \mathcal{C}(c^1,\dots ,c^s)$. We have 
    \begin{align}
        &E(I^1, \dots, I^p)  \nonumber \\
        =&~ \Ex{\prod_{\ell=1}^p M[I^{\ell}, I^{\ell + 1}] } \nonumber \\
		=&~ \sum_{\pi_j, \sigma_j \in \Sym_{q/2}}   \Ex{\prod_{\ell=1}^p \prod_{m = 1}^{q/d} A[(\pi ( I^{\ell}))_{m; d/2}, (\pi ( I^{\ell + 1}))_{m; d/2} ]}  \nonumber \\
        =&~  
\bigg(
        \frac{\prod_{\ell} \prod_s (c_{\ell}^{s}!)^{2}}
        {((q/2)!)^{2p}} \bigg) \cdot
        \sum\limits_{(J^{\ell},K^{\ell} \in \orbit{I^{\ell}})_{\ell \in [p]}} 
        \Ex{\prod_{\ell=1}^p \prod_{m = 1}^{q/d} A[J^{\ell}_{m; d/2}, K^{\ell+1}_{m; d/2} ] } 
    \Eqlab{fixed:mult}
\end{align}

Thus, $E(I^1, \dots, I^p)$ is bounded by the number of choices for $J^1, \dots, J^{p}, K^1, \dots, K^{p}$ such that $J^{\ell}, K^{\ell} \in \orbit{I^{\ell}}$ for each $\ell \in [p]$, and 
$\Ex{\prod_{\ell=1}^p \prod_{m = 1}^{q/d} A[J^{\ell}_{m; d/2}, K^{\ell+1}_{m; d/2} ]}$ is nonzero. 

Given $J^1, \dots, J^p$ and $K^1, \dots, K^p$, consider the $(pq/d)$-tuple $T$ where each coordinate is indexed by $(\ell, m)_{\ell \in [p], m \in [q/d]}$ and has a $d$-tuple $T_{\ell, m} := (J^{\ell}_{m;d/2}) \oplus (K^{\ell + 1}_{m;d/2}) \in \RR^{d}$ as a value. 
Note that $\sum_{\ell, m} \mi{T_{\ell, m}}) = (2o_1, \dots, 2o_n)$ where $o_r$ is the number of occurences of $r \in [n]$ in $(pq/2)$-tuple $\oplus_{\ell = 1}^p I^{\ell}$. 
The fact that $\Ex{\prod_{\ell=1}^p \prod_{m = 1}^{q/d} A[j_{m; d/2}, k_{m; d/2} ]} \neq 0$ means that every $d$-tuple occurs even number of times in $T$. 

We count the number of $(pq/d)$-tuples $T = (T_{\ell, m})_{\ell \in [p], m \in [q]}$ that $\sum_{\ell, m} \mi{T_{\ell, m}} = (2o_1, \dots, 2o_n)$ and every $d$-tuple occurs an even number of times. 
Let $Q = (Q_1, \dots, Q_{pq/2d})$, $R = (R_1, \dots, R_{pq/2d})$ be two $(pq/2d)$-tuples of $d$-tuples where for every $d$-tuple $P$, the number of occurences of $P$ is the same in $Q$ and $R$, and $\sum_{\ell = 1}^{pq/2d} \mi{Q_\ell} = \sum_{\ell = 1}^{pq/2d} \mi{R_\ell} = (o_1, \dots, o_n)$. At most $2^{pq/d}$ tuples $T$ can be made by {\em interleaving} $Q$ and $R$ --- for each $(\ell, m)$, choose $T_{\ell, m}$ from the first unused $d$-tuple in either $Q$ or $R$. Furthermore, every tuple $T$ that meets our condition can be constructed in this way. 

Due to the condition 
$\sum_{\ell = 1}^{pq/2d} \mi{Q_\ell} = (o_1, \dots, o_n)$, the number of choices for $Q$ is at most the number of different ways to permute $I^1 \oplus \dots \oplus I^p$, which is at most $(pq/2)! / \prod_{m \in [s]} (\bar{c}^m)!$, where $\bar{c}^m := \sum_{\ell\in [p]} c^{m}_{\ell}$ for $m \in [s]$. For a fixed choice of $Q$, there are at most $(pq/2d)!$ choices of $R$. Therefore, the number of choices for 
$(J^{\ell},K^{\ell} \in \orbit{I^{\ell}})_{\ell \in [p]}$ with nonzero expected value is at most 
\[
2^{pq/d} \cdot \frac{(pq/2)!}{\prod_{m \in [s]} (\bar{c}^m)!} \cdot (pq/2d)! = 2^{O(pq)} \cdot \frac{(pq)^{1/2 + 1/2d}}{\prod_{m \in [s]} (\bar{c}^m)!}.
\]
Combining with~\Eqref{fixed:mult},  
\[
E(I^1, \dots, I^p) \leq \bigg( 
2^{O(pq)} \frac{(pq)^{1/2 + 1/2d}}{\prod_{m \in [s]} (\bar{c}^m)!} \bigg) \cdot \bigg(
        \frac{\prod_{\ell} \prod_s (c_{\ell}^{s}!)^{2}}
        {((q/2)!)^{2p}} \bigg) 
\leq 2^{O(pq)} \cdot \frac{p^{1/2 + 1/2d}}{q^{1/2 - 1/2d}} 
\cdot \prod_{\ell} \prod_s c_{\ell}^{s}!
\]
as desired. 
\end{proof}

\begin{lemma}
\lemlab{num:dist}
    For all $I^{1},\dots ,I^{p}\in [n]^{q/2}$, if $E(I^1, \dots, I^p) \neq 0$, $\numdist{I^1, \dots, I^p} \leq \frac{pq}{4} + \frac{q}{2}$. 
\end{lemma}

\begin{proof}
Note that $E(I^1, \dots, I^p) \neq 0$ 
    implies that 
there exist $J^1, \dots, J^p, K^1, \dots, K^p$ such that $J^{\ell}, K^{\ell} \in \orbit{I^{\ell}}$ and 
every $d$-tuple occurs exactly even number of times in 
$((J^{\ell}_{m;d/2}) \oplus (K^{\ell + 1}_{m;d/2}))_{\ell \in [p], m \in [q/d]}$. Consider the graph $G = (V, E)$ defined by 
    \begin{align*}
        V &:= 
        \bigcup_{\ell\in [p]}\bigcup_{k\in [q/2]} \brc{I^\ell_k} \\
        E &:= \bigcup_{m\in [q/2]} \brc{\brc{J^1_{m},K^2_{m}},\brc{J^2_{m},K^3_{m}},\dots ,
        \brc{J^{p}_{m},K^1_{m}}}.
    \end{align*}
    The even multiplicity condition implies that every element in 
    $E$ has even multiplicity and consequently $|E|\leq pq/4$. We next show that $E$ is the union of $q/2$ paths. To this end, we construct $G^1\in\orbit{I^1},\dots ,
    G^\ell\in\orbit{I^\ell}$ as follows:
    
    \begin{enumerate}
        \item Let $G^2:=K^2$
        
        \item For $3\leq \ell \leq p$ do:
        \begin{enumerate}[\quad]
            \item Since $G^\ell\in \orbit{J^\ell}$, there exists $\pi\in \Sym_{q/2}$ s.t. 
            $\pi ( J^\ell) = G^{\ell}$.
            \item Let $G^{\ell + 1} := \pi ( K^{\ell+1})$.
        \end{enumerate}
    \end{enumerate}
    
    \noindent
    We observe that by construction,
    \begin{align*}
        &\quad 
        \bigcup_{m\in [q/2]} \brc{\brc{J^1_{m},G^2_{m}},\brc{G^2_{m},G^3_{m}},\dots ,
        \brc{G^{p}_{m},G^1_{m}}} \\
        &= 
        \bigcup_{m\in [q/2]} \brc{\brc{J^1_{m},K^2_{m}},\brc{J^2_{m},K^3_{m}},\dots ,
        \brc{J^{p}_{m},K^1_{m}}}
        = 
        E
    \end{align*}
    which establishes that $E$ is a union of $q/2$ paths. 
    
    Now since $E$ is the union of $q/2$ paths $G$ has at most $q/2$ connected components, and one needs to add at most $q/2 -1$ edges make it connected, we have $|V|\leq |E| + (q/2 - 1) + 1 \leq pq/4 + q/2$. But $\numdist{I^1,\dots ,I^p} = |V|$, which completes the proof. 
\end{proof}

Finally, $\Ex{\Tr{M^{p}}}$ can be bounded as follows. 
\begin{align*}
    &\quad \Ex{\Tr{M^{p}}}  \\
    &= 
    \sum\limits_{I^1,\dots ,I^p \in [n]^{q/2}}
    E(I^1, \dots, I^p) \\
    &= 
    \sum\limits_{s\in [pq/4 + q/2]} 
    ~\sum\limits_{\numdist{I^1,\dots ,I^p}=s}
    E(I^1, \dots, I^p)
    &&\text{(by \lemref{num:dist})}\\
    &= 
    \sum\limits_{s\in [pq/4 + q/2]} 
    ~\sum\limits_{c^1,\dots ,c^s\in [q/2]^p} 
    ~\sum\limits_{(I^1,\dots ,I^p) \in \mathcal{C}(c^1\dots c^s)}
    E(I^1, \dots, I^p) \\
    &= 
    \sum\limits_{s\in [pq/4 + q/2]} 
    ~\sum\limits_{c^1,\dots ,c^s\in [q/2]^p} 
    ~\sum\limits_{(I^1,\dots ,I^p) \in \mathcal{C}(c^1\dots c^s)}
    E(I^1, \dots, I^p) \\
    &\leq 
    \sum\limits_{s\in [pq/4 + q/2]} 
    ~\sum\limits_{c^1,\dots ,c^s\in [q/2]^p} \\
    &~\sum\limits_{(I^1,\dots ,I^p) \in \mathcal{C}(c^1\dots c^s)}
    2^{O(pq)}\,\frac{p^{(1/2 + 1/2d)pq}}{q^{(1/2 - 1/2d)pq}}\,
    \prod_{\ell\in [p]} c_\ell^{1}!\dots c_\ell^{s}! 
    &&\text{(by \lemref{fixed:mult})} \\
    &\leq 
    \sum\limits_{s\in [pq/4 + q/2]} 
    2^{O(pq)}\,\frac{n^s}{s!}\,p^{(1 + 1/2d)pq}q^{pq/2d}
    &&\text{(by \obsref{num:guesses})} \\
    &\leq 
    \sum\limits_{s\in [pq/4 + q/2]} 
    2^{O(pq)}\,\frac{n^{pq/4+q/2}}{s!\,q^{pq/4+q/2-s}}\,p^{(1/2 + 1/2d)p1}q^{(1/2 - 1/2d)pq}
    &&(\text{assuming } q\leq n) \\
    &\leq 
    \sum\limits_{s\in [pq/4 + q/2]} 
    2^{O(pq)}\,\frac{n^{pq/4+q/2}\,p^{(1+1/2d)pq}}{q^{(1/4-1/2d)pq}} \\
    &\leq 
    2^{O(pq)}\,\frac{n^{pq/4+q/2}\,p^{(1+1/2d)pq}}{q^{(1/4-1/2d)pq}}.
\end{align*}
Choose $p$ to be even and let $p=\Theta(\log n)$. 
Applying Markov inequality shows that with high probability, 
\[
\pth{\hssos{f^{q/d}}}^{d/q}
\leq 
(\norm{2}{M})^{d/q} \leq (\Ex{\Tr{M^{p}}})^{d/pq} 
= O_d\pth{\frac{n^{d/4} \cdot (\log n)^{\,\,d + 1/2} }{q^{\,d/4 - 1/2}}}.
\]

Thus we obtain
\begin{theorem}
\thmlab{dten:even:ub}
    For even $d$, let $\tenAA\in \Re^{[n]^d}$ be a $d$-tensor with i.i.d. 
    $\pm 1$ entries. Then for any even $q$ such that $q\leq n$, 
    we have that with probability $1-n^{\Omega(1)}$, 
    \[
        \frac{\sos{q}{\tenAA(x)}}{\fmax{\tenAA}} ~~\leq ~~
        \pth{\frac{\widetilde{O}(n)}{q}}^{d/4-1/2} .
    \]
\end{theorem}

\paragraph{Remark.}
For the special case where $q=d$, we prove a stronger upper bound, namely 
\[
    \frac{\sos{q}{\tenAA(x)}}{\fmax{\tenAA}} ~~\leq ~~
    \pth{\frac{O(n)}{q}}^{d/4-1/2}~,
\]
the proof of which is implicit in the proofs of \lemref{mineig:quo} and 
\lemref{quo:op}.

\section{Proof of SoS Lower Bound in~\thmref{q-tensor:informal}}
\seclab{qten:lb}

For even $q$, let $\tenAA \in \Re^{[n]^{q}}$ be a $q$-tensor with i.i.d. $\pm 1$ entries and let $A \in \Re^{[n]^{q/2} \times [n]^{q/2}}$ be the matrix flattening of $\tenAA$, \ie 
$A[I,J] = \tenA{I\oplus J}$ (recall that $\oplus$ denotes tuple concatenation). Also let $f(x) := \tenAA(x) = \iprod{\tenAA}{x^{\otimes q}}$. 
This section proves the lower bound in~\thmref{q-tensor:informal}, by constructing a moment matrix $\sfM$ that is positive semidefinite, SoS-symmetric, $\Tr{\sfM} = 1$, and 
$\langle A, \sfM \rangle \geq 2^{-O(q)} \cdot \frac{n^{q/4}}{q^{q/4}}$.
In~\secref{wigner:subsec}, we construct the matrix $\hW$ that acts as a SoS-symmetrized identity matrix. 
The moment matrix $\sfM$ is presented in~\secref{subsec:random_lower_m}. 

\subsection{Wigner Moment Matrix}
\seclab{wigner:subsec}

In this section, we construct an SoS-symmetric and positive semidefinite matrix $\hW \in \Re^{\degmindex{q/2} \times \degmindex{q/2}}$ such that $\lambda_{\min}(\hW) / \Tr{\hW} \geq 1/(2^{q + 1} \cdot |\degmindex{q/2}|)$, i.e. the ratio of the minimum eigenvalue to the average eigenvalue is at least $1/2^{q+1}$. 
\begin{theorem}
\thmlab{identity}
For any positive integer $n$ and any positive even integer $q$, there exists a matrix $\hW \subseteq \Re^{\degmindex{q/2} \times \degmindex{q/2}}$ that satisfies the following three properties: (1) $\hW$ is degree-$q$ SoS symmetric. (2) The minimum eigenvalue of $\hW$ is at least $\frac{1}{2}$. (3) Each entry of $\hW$ is in $[0, 2^q]$.
\end{theorem}
\thmref{identity} is proved by explicitly constructing independent random variables $x_1, \dots, x_n$ such that for any $n$-variate polynomial $p(x_1, \dots, x_n)$ of degree at most $\frac{q}{2}$, $\E[p^2]$ is bounded away from $0$. The proof consists of three parts. The first part shows the existence of a desired distribution for one variable $x_i$. The second part uses induction to prove that $\E[p^2]$ is bounded away from $0$. The third part constructs $\hW \subseteq \Re^{\degmindex{q/2} \times \degmindex{q/2}}$ from the distribution defined.

\paragraph{Wigner Semicircle Distribution and Hankel Matrix.}
Let $k$ be a positive integer.
In this part, the rows and columns of all $(k + 1) \times (k + 1)$ matrices are indexed by $\{0, 1, \dots, k \}$. 
Let $T$ be a $(k + 1) \times (k + 1)$ matrix where $T[i,j] = 1$ if $|i - j| = 1$ and $T[i, j] = 0$ otherwise. 
Let $e_0 \in \Re^{k + 1}$ be such that $(e_0)_0 = 1$ and $(e_0)_i = 0$ for $1 \leq i \leq k$. 
Let $R \in \Re^{(k + 1) \times (k + 1)}$ be defined by $R := [e_0, Te_0, T^2e_0, \dots, T^{k} e_0]$. Let $R_0, \dots, R_k$ be the columns or $R$ so that $R_i = T^{i} e_0$. 
It turns out that $R$ is closely related to the number of ways to consistently put parantheses. Given a string of parantheses `(' or `)', we call it {\em consistent} if any prefix has at least as many `(' as `)'. For example, ((())( is consistent, but ())(( is not. 

\begin{claim}
$R[i, j]$ is the number of ways to place $j$ parantheses `('  or `)' consistently so that there are $i$ more `(' than `)'. 
\claimlab{parenthesis}
\end{claim}
\begin{proof}
We proceed by the induction on $j$. When $j = 0$, $R[0, 0] = 1$ and $R[i, 0] = 0$ for all $i \geq 1$. 
Assume the claim holds up to $j - 1$. By the definition $R_j = TR_{j - 1}$. 
\begin{itemize}
\item For $i = 0$, the last parenthesis must be the close parenthesis, so the definition $R[0, j] = R[1, j - 1]$ still measures the number of ways to place $j$ parantheses with equal number of `(' and `)'. 
\item For $i = k$, the last parenthesis must be the open parenthesis, so the definition $R[k, j] = R[k - 1, j - 1]$ still measures the number of ways to place $j$ parantheses with $k$ more `('. 
\item For $0 < i < k$, the definition of $R$ gives $R[i, j] = R[i - 1, j - 1] + R[i + 1, j - 1]$. 
Since $R[i - 1, j]$ corresponds to plaincg `)' in the $j$th position and $R[i + 1, j]$ corresponds to placing `(' in the $j$th position, $R[i, j]$ still measures the desired quantity. 
\end{itemize}
This completes the induction and proves the claim.
\end{proof}

Easy consequences of the above claim are (1) $R[i, i] = 1$ for all $0 \leq i \leq k$, and $R[i, j] = 0$ for $i > j$, and (2) $R[i, j] = 0$ if $i + j$ is odd, and $R[i, j] \geq 1$ if $i \leq j$ and $i + j$ is even. 

Let $H := (R^T) R$. Since $R$ is upper triangular with $1$'s on the main diagonal, $H = (R^T)R$ gives the unique Cholesky decomposition, so $H$ is positive definite. It is easy to see that $H[i, j] = \langle R_i, R_j \rangle$ is the total number of ways to place $i + j$ parantheses consistently with the same number of `(' and `)'. Therefore, $H[i, j] = 0$ if $i + j$ is odd, and if $i + j$ is even (let $l := \frac{i + j}{2}$), $H[i, j]$ is the $l$th Catalan number $C_{l} := \frac{1}{l + 1} \binom{2l}{l}$. 
In particular, $H[i, j] = H[i', j']$ for all $i + j = i' + j'$. Such $H$ is called a {\em Hankel matrix}. 

Given a sequence of $m_0 = 1, m_1, m_2, \dots$ of real numbers, the {\em Hamburger moment problem} asks whether there exists a random variable $W$ supported on $\Re$ such that $\E[W^i] = m_i$. It is well-known that there exists a unique such $W$ if for all $k \in \NN$, the Hankel matrix $H \in \Re^{(k + 1) \times (k + 1)}$ defined by $H[i, j] := \E[W^{i + j}]$ is positive definite~\cite{Simon98}. Since our construction of $H \in \Re^{(k + 1) \times (k + 1)}$ ensures its positive definiteness for any $k \in \NN$, there exists a unique random variable $W$ such that $\E[W^i] = 0$ if $i$ is odd, $\E[W^i] = C_{\frac{i}{2}}$ if $i$ is even. It is known as the {\em Wigner semicircle distribution} with radius $R = 2$. 

\begin{remark}
Some other distributions (e.g., Gaussian) will give an asymptotically weaker bound. Let $G$ be a standard Gaussian random variable. The quantitative difference comes from the fact that $\E[W^{2l}] = C_{l} = \frac{1}{l + 1} \binom{2l}{l} \leq 2^{l}$ while $\E[G^{2l}] = (2l - 1)!! \geq 2^{\Omega(l \log l)}$. 
\end{remark}

\paragraph{Multivariate Distribution.}
Fix $n$ and $q$. Let $k = \frac{q}{2}$. 
Let $H \in \Re^{(k + 1) \times (k + 1)}$ be the Hankel matrix defined as above, and $W$ be a random variable sampled from the Wigner semicircle distribution. 
Consider $x_1, \dots, x_n$ where each $x_i$ is an independent copy of $\frac{W}{N}$ for some large number $N$ to be determined later. 
Our $\hW$ is later defined to be $\hW[\alpha, \beta] = \E[x^{\alpha + \beta}] \cdot N^{q}$ so that the effect of the normalization by $N$ is eventually cancelled,
but large $N$ is needed to prove the induction that involves non-homogeneous polynomials. 

We study $\E[p(x)^2]$ for any $n$-variate (possibly non-homogeneous) polynomial $p$ of degree at most $k$. For a multivarite polynomial $p = \sum_{\alpha \in \udmindex{k}} p_{\alpha} x^{\alpha}$, 
define $\ell_2$ norm of $p$ to be $\norm{\ell_2}{p} := \sqrt{\sum_{\alpha} p_{\alpha}^2}$. 
For $0 \leq m \leq n$ and $0 \leq l \leq k$, let 
$\sigma(m, l) := \inf_p \E [p(x)^2]$ where the infimum is taken over polynomials $p$ such that $\norm{\ell_2}{p} = 1$, $\deg(p) \leq l$, and $p$ depends only on $x_1, \dots, x_m$. 

\begin{lemma}
There exists $N := N(n, k)$ such that $\sigma(m, l) \geq \frac{(1 - \frac{m}{2n}) }{N^{2l}}$ for all $0 \leq m \leq n$ and $0 \leq l \leq k$. 
\lemlab{lem:induction}
\end{lemma}
\begin{proof}
We prove the lemma by induction on $m$ and $l$. 
When $m = 0$ or $l = 0$, $p$ becomes the constant polynomial $1$ or $-1$, so $\E[p^2] = 1$. 

Fix $m, l > 0$ and a polynomial $p = p(x_1, \dots, x_m)$ of degree at most $l$. 
Decompose $p = \sum_{i = 0}^l p_i x_m^i$ where each $p_i$ does not depend on $x_m$. The degree of $p_i$ is at most $l - i$. 
\[
\E[p^2]  = 
\E[(\sum_{i = 0}^l p_i x_m^i)^2] 
 = \sum_{0 \leq i, j \leq l} \E[p_i p_j] \E[x_m^{i + j}]. \nonumber
\]

Let $\Sigma = \diag(1, \frac{1}{N}, \dots, \frac{1}{N^{l}}) \in \Re^{(l+1) \times (l+1)}$. 
Let $H_l \in \Re^{(l + 1) \times (l + 1)}$ be the submatrix of $H$ with the first $l + 1$ rows and columns. 
The rows and columns of $(l + 1) \times (l + 1)$ matrices are still indexed by $\{0, \dots, l \}$. 
Define $R_l \in \Re^{(l + 1) \times (l + 1)}$ similarly from $R$, and $r_t$ ($0 \leq t \leq l$) be the $t$th column of $(R_l)^T$. 
Note $H_l = (R_l)^T R_l = \sum_{t = 0}^{l}r_t r_t^T$. 
Let $H' = \Sigma H_l \Sigma$ such that $H'[i, j] = \E[x_m^{i + j}]$. 
Finally, let $P \in \Re^{(l + 1) \times (l + 1)}$ be defined such that $P[i, j] := \E[p_{i} p_{j}]$. 
Then $\E[p^2]$ is equal to 
\begin{align*}
&\quad \Tr{PH'} = \Tr{P \Sigma H_l \Sigma} 
= \Tr{P \Sigma (\sum_{t=0}^{l} r_t r_t^T ) \Sigma} \\
&= \sum_{t=0}^{l} \E[(p_{t}  \frac{1}{N^t} +  p_{t + 1} \frac{(r_t)_{t + 1}}{N^{t + 1}} + \dots +  p_{l}\frac{(r_t)_{l}}{N^{l}})^2],
\end{align*}
where the last step follows from the fact that $(r_t)_j = 0$ if $j < t$ and $(r_t)_t = 1$. 
Consider the polynomial
\[
q_t := p_{t}  \frac{1 }{N^{t}} +  p_{t + 1} \frac{(r_t)_{t + 1}}{N^{t + 1}} + \dots + p_{l}\frac{(r_t)_{l}}{N^{l}}.
\]
Since $p_i$ is of degree at most $l - i$, $q_t$ is of degree at most $l - t$. 
Also recall that each entry of $R$ is bounded by $2^k$.
By the triangle inequality, 
\[
\norm{\ell_2}{q_t} \geq \frac{1}{N^{t}} \bigg( \norm{\ell_2}{p_t } - \big( \norm{\ell_2}{p_{t+1}} \frac{(r_t)_{t+1}}{N} + \dots + \norm{\ell_2}{p_l}\frac{(r_t)_{l}}{N^{l - t}} \big)  \bigg)\geq \frac{1}{N^{t}} \bigg( \norm{\ell_2}{p_t} -\frac{ k 2^k }{N}  \bigg),
\]
and 
\[
\norm{\ell_2}{q_t}^2 \geq \frac{1}{N^{2t}} \bigg( \norm{\ell_2}{p_t}^2 - \frac{2 k 2^k }{N}  \bigg).
\]
Finally, 
\begin{align*}
\E[p^2] &= \sum_{t=0}^{l} \E[q_t^2] \\
&\geq \sum_{t=0}^{l} \sigma (m - 1, l - t) \cdot \norm{\ell_2}{q_t}^2 \\
& \geq \sum_{t=0}^{l} \sigma (m - 1, l - t) \cdot \frac{1}{N^{2t}} \bigg( \norm{\ell_2}{p_t}^2 -\frac{ 2 k 2^k }{N}  \bigg) \\
&\geq \sum_{t=0}^{l} \frac{(1 - \frac{m - 1}{2n})}{N^{2l - 2t}} \cdot \frac{1}{N^{2t}} \cdot \bigg( \norm{\ell_2}{p_t}^2 -\frac{ 2 k 2^k }{N}  \bigg) \\ 
&= \frac{(1 - \frac{m - 1}{2n})}{N^{2l}} \cdot \sum_{t=0}^{l} \bigg( \norm{\ell_2}{p_t}^2 -\frac{ 2 k 2^k }{N}  \bigg) \\
&\geq \frac{(1 - \frac{m - 1}{2n})}{N^{2l}} \cdot 
\big( 1 - \frac{2K^2 2^k}{N} \big). 
\end{align*}
Take $N := 4nK^2 2^k$ so that 
$
\big( 1 - \frac{m - 1}{2n} \big) \cdot \big( 1 - \frac{2K^2 2^k}{N} \big) \geq 
1 - \frac{m - 1}{2n} - \frac{2K^2 2^k}{N} =  1 - \frac{m}{2n}.
$
This completes the induction and proves the lemma. 
\end{proof}

\paragraph{Construction of $\hW$.}
We now prove~\thmref{identity}. 
Given $n$ and $q$, let $k = \frac{q}{2}$, and consider random variables $x_1, \dots, x_n$ above. Let $\hW \in \Re^{\degmindex{k} \times \degmindex{k}}$ be such that for any $\alpha, \beta \in \degmindex{k}$, $\hW[\alpha, \beta] = \E[x^{\alpha + \beta}] \cdot N^{2k}$. 
By definition, $\hW$ is degree-$q$ SoS symmetric. 
Since each entry of $\hW$ corresponds to a monomial of degree exactly $q$ and each $x_i$ is drawn independently from the Wigner semicircle distribution, each entry of $\hW$ is at most the $\frac{q}{2}$th Catalan number $C_{\frac{q}{2}} \leq 2^q$. 
For any unit vector $p = (p_S)_{S \in \degmindex{k}} \in \Re^{\degmindex{k}}$, \lemref{lem:induction} shows $p^T \hW p = \E[p^2] \cdot N^{2k} \geq \frac{1}{2}$ where $p$ also represents a degree-$k$ homogeneous polynomial $p(x_1, \dots, x_n) = \sum_{\alpha \in \binom{[n]}{k}} p_{\alpha} x^{\alpha}$. Therefore, the minimum eigenvalue of $\hW$ is at least $\frac{1}{2}$.

\subsection{Final Construction}
\seclab{subsec:random_lower_m}
For even $d$, let $\tenAA \in \Re^{[n]^{q}}$ be a $q$-tensor with i.i.d. $\pm 1$ entries and let $A \in \Re^{[n]^{q/2} \times [n]^{q/2}}$ be the matrix flattening of $\tenAA$, \ie 
$A[I,J] = \tenA{I\oplus J}$ (recall that $\oplus$ denotes tuple concatenation). Also let $f(x) := \tenAA(x) = \iprod{\tenAA}{x^{\otimes q}}$. 
Our lower bound on $\fmax{f}$ by is proved by constructing a moment matrix $\sfM \in \RR^{[n]^{q/2} \times [n]^{q/2}}$ that satisfies 
\begin{itemize}
\item $\Tr{\sfM} = 1$.
\item $\sfM\succeq 0$. 
\item $\sfM$ is SoS-symmetric.
\item $\langle A, \sfM \rangle ~~\geq ~~2^{-O(q)}\cdot n^{q/4}/q^{q/4}$,
\end{itemize}
where $A \in \Re^{[n]^{q/2} \times [n]^{q/2}}$ is any matrix representation of $f$ (SoS-symmetry of $\sfM$ ensures $\langle A, \sfM \rangle$ does not depend on the choice of $A$). 

Let $\sfA$ be the SoS-symmetric matrix such that for any $I = (i_1, \dots, i_{q/2})$ and $J = (j_1, \dots, j_{q/2})$, 
\[
\sfA[I, J] = 
\begin{cases}
\frac{f_{\mi{I}+\mi{J}}}{q!}, & \mbox {if } i_1, \dots, i_{q/2}, j_1, \dots, j_{q/2} \mbox{ are all distinct.} \\
 0 & \mbox {otherwise.}
\end{cases}
\]
We bound $\norm{2}{\sfA}$ in two steps. 
Let $\hA_Q \in \Re^{\degmindex{q/2} \times \degmindex{q/2}}$ be the {\em quotient matrix} of $\sfA$ defined by
\[
\hA_Q[\beta, \gamma] := \sfA[I, J] \cdot \sqrt{ |\orbit{\beta}|\cdot |\orbit{\gamma}|},
\]
where $I, J \in [n]^{q/2}$ are such that $\beta = \mi{I}, \gamma = \mi{J}$.

\begin{lemma}
\lemlab{mineig:quo} With high probability,  
    $
	    \| \hA_Q \|_2 \leq 
	    ~ 2^{O(q)} \cdot \frac{n^{q/4}}{q^{q/4}}
    $.
\end{lemma}

\begin{proof}
	Consider any $y\in \Re^{\degmindex{q/2}}$ s.t. $\|y\| = 1$. Since
	\begin{align*}
		y^T\cdot \hA_Q \cdot y 
		&= 
		\sum_{\beta+ \gamma \,\leq \,\1} 
		\ha_Q[\beta, \gamma] \cdot y_{\beta}\cdot y_{\gamma} \\
        &= 
        \sum_{\beta+ \gamma \,\leq \,\1} 
        y_{\beta}\cdot y_{\gamma} 
        \sum_{
        \substack{
        \mi{I} + \mi{J} \\ = \beta + \gamma
        } 
        } 
        A[I, J] \cdot \frac{\sqrt{|\orbit{\beta}||\orbit{\gamma}|}}{|\orbit{\beta + \gamma}|} \\ 
        &=
        \sum_{I, J \in [n]^{q/2}} A[I, J]
        \sum_{
        \substack{
        \beta+ \gamma \,\leq \,\1 \\ \beta + \gamma =\\ \mi{I} + \mi{J} 
        } 
        }
        \frac{\sqrt{|\orbit{\beta}||\orbit{\gamma}|}}{|\orbit{\beta + \gamma}|} \cdot 
        y_{\beta}\cdot y_{\gamma}  
	\end{align*}
	So $y^T \cdot \hA_Q \cdot y$ is a sum of independent random variables 
    \[
        \sum_{I, J \in [n]^q} A[I, J] \cdot c_{I, J}
    \]
    where each $A[I, J]$ is independently sampled from the Rademacher distribution and 
    \[
        c_{I, J} := 
        \sum_{
        \substack{
        \beta+ \gamma \,\leq \,\1 \\ \beta + \gamma =\\ \mi{I} + \mi{J} 
        } 
        } 
        \frac{\sqrt{|\orbit{\beta}||\orbit{\gamma}|}}{|\orbit{\beta + \gamma}|}  
        \cdot y_{\beta}\cdot y_{\gamma} \,.
    \]
    Fix any $I,J\in [n]^{q/2}$ and let $\alpha := \mi{I}+\mi{J}$. By Cauchy-Schwarz, 
    \begin{equation}
    \Eqlab{entry:var}
        c_{I, J}^2 
        \quad \leq \quad
        \bigg( 
        \sum_{
        \beta + \gamma = \alpha
        }
        \frac{|\orbit{\beta}||\orbit{\gamma}|}{|\orbit{\alpha}|^2} 
        \bigg) 
        \cdot 
        \big( 
        \sum_{
        \beta + \gamma = \alpha
        }
        y_{\beta}^2 \cdot y_{\gamma}^2 
        \big)
        \quad \leq \quad
        \frac{2^{O(q)}}{|\orbit{\alpha}|}\cdot 
        \sum_{
        \beta + \gamma = \alpha
        } 
        y_{\beta}^2 \cdot y_{\gamma}^2 
        \quad =:\quad
        c^2_{\alpha} \,,
    \end{equation}
    since there are at most $2^{O(q)}$ choices of $\beta$ and $\gamma$ with $\beta + \gamma = \alpha$, 
    and $|\orbit{\beta}| \cdot |\orbit{\gamma}| \leq |\orbit{\alpha}|$. 
    Therefore, $y^T \cdot \hA_Q \cdot y$~ is the sum of independent random variables that are 
    centred and always lie in the interval $[-1, +1]$. Furthermore, by \Eqref{entry:var}, the total 
    variance is 
    \[
        \sum_{I,J\in [n]^{q/2}} c_{I,J}^2 
        ~~\leq ~~
        \sum_{\alpha\in \degmindex{q}} c_{\alpha}^2 \cdot |\orbit{\alpha}|
        ~~\leq ~~
        2^{O(q)} \cdot \sum_{\beta, \gamma \in \degmindex{q/2}} y_{\beta}^2 \cdot y_{\gamma}^2
        ~~= ~~
        2^{O(q)} \cdot \big(\sum_{\beta \in \degmindex{q/2}} y_{\beta}^2 \big)^2
        ~~= ~~
        2^{O(q)}
    \]
    The claim then follows from combining standard concentration bounds with a union bound over 
    a sufficiently fine net of the unit sphere in 
    $|\degmindex{q/2}| \leq 2^{O(q)} \cdot \frac{n^{q/2}}{q^{q/2}}$~ dimensions. 
\end{proof}

\begin{lemma}
\lemlab{quo:op}
    For any SoS-symmetric $\sfA \in \Re^{[n]^{q/2}\times [n]^{q/2}}$, 
    $\norm{2}{\sfA} \leq \norm{2}{\hA_Q}$.
\end{lemma}

\begin{proof}
    For any $u,v\in \Re^{[n]^{q/2}} \,s.t.~\|u\|=\|v\|=1$, we have 
    \begin{align*}
        &\quad 
        u^T \sfA v \\
        &= \sum\limits_{I,J\in [n]^{q/2}} \sfA[I, J] u_I v_J \\ 
        &= 
        \sum\limits_{I,J\in [n]^{q/2}} 
        \frac{
        \hA_Q[\mi{I}, \mi{J}]
        }
        {
        \sqrt{|\orbit{I}|\,|\orbit{J}|}
        } \cdot u_I v_J \\
        &=
        \sum\limits_{\alpha, \beta \in \degmindex{q/2}} 
        \frac{\sfA[\alpha , \beta]}{\sqrt{|\orbit{\alpha}|\,|\orbit{\beta}|}}
        \iprod{u\restrict{\orbit{\alpha}}}{\one}\iprod{v\restrict{\orbit{\beta}}}{\one} 
        \\
        &= 
        a^T\hA_Q ~b
        \qquad\text{where } 
        a_{\alpha} := \frac{\iprod{u\restrict{\orbit{\alpha}}}{\one}}{\sqrt{|\orbit{\alpha}|}},~
        b_{\alpha} := \frac{\iprod{v\restrict{\orbit{\alpha}}}{\one}}{\sqrt{|\orbit{\alpha}|}} \\
        &\leq 
        \norm{2}{\hA_Q} \|a\|\cdot \|b\| \\
        &= 
        \norm{2}{\hA_Q} \sqrt{\sum\limits_{\alpha\in \degmindex{q/2}} 
        \frac{\iprod{u\restrict{\orbit{\alpha}}}{\one}^2}{|\orbit{\alpha}|}}
        \sqrt{\sum\limits_{\alpha\in \degmindex{q/2}} 
        \frac{
        \iprod{v\restrict{\orbit{\alpha}}}{\one}^2}{|\orbit{\alpha}|}} \\
        &\leq 
        \norm{2}{\hA_Q} \sqrt{\sum\limits_{\alpha\in \degmindex{q/2}} 
        \|u\restrict{\orbit{\alpha}}\|^2 }
        \sqrt{\sum\limits_{\alpha\in \degmindex{q/2}} 
        \|u\restrict{\orbit{\alpha}}\|^2 }
        &&\text{(by Cauchy-Schwarz)} \\
        &\leq 
        \norm{2}{\hA_Q} \|u\|\cdot \|v\| = \norm{2}{\hA_Q}.
    \end{align*}
\end{proof}

The above two lemmas imply that $\norm{2}{\sfA} \leq 
\| \hA_Q \|_2 \leq 
	    ~ 2^{O(q)} \cdot \frac{n^{q/4}}{q^{q/4}}$. 
Our moment matrix $\sfM$ is defined by
\[
\sfM := \frac{1}{c_1} \bigg( \frac{1}{c_2} \cdot \frac{q^{3q/4}}{n^{3q/4}} \sfA + \frac{\sfW}{n^{q/2}} \bigg),
\]
where $\sfW$ is the direct extension of $\hW$ constructed in~\thmref{identity} --- $\sfW[I, J] := \hW[\mi{I}, \mi{J}]$ for all $I, J \in [n]^{q/2}$, and $c_1, c_2 = 2^{\Theta(q)}$ that will be determined later. 

We first consider the trace of $M$. The trace of $\sfA$ is $0$ by design, and the trace of $\sfW$ is $n^{q/2} \cdot 2^{O(q)}$. Therefore, the trace of $\sfM$ can be made $1$  by setting $c_1$ appropriately. Since both $\sfA$ and $\sfW$ are SoS-symmetric, so is $\sfM$. 
Since $\E[\sfW, A] = 0$ and for each $I, J \in [n]^{q/2}$ with $i_1, \dots, i_{q/2}, j_1, \dots, j_{q/2}$ all distinct we have $\E[\sfA[I, J] A[I, J]] = \frac{1}{q!}$, with high probability 
\[
\langle A, \sfM \rangle  = \frac{1}{c_1} \cdot  \langle A,  \bigg( \frac{1}{c_2}  \cdot \frac{q^{3q/4}}{n^{3q/4}} \sfA + \frac{\sfW}{n^{q/2}} \bigg)  \rangle
\geq 2^{O(-q)} \cdot \frac{q^{3q/4}}{n^{3q/4}} \cdot \frac{n^q}{q^q} = 2^{O(-q)} \cdot \frac{n^{q/4}}{q^{q/4}}.
\]
It finally remains to show that $\sfM$ is positive semidefinite. Take an arbitrary vector $v \in \RR^{[n]^{q/2}}$, and let 
\[
p = \sum_{\alpha \in \degmindex{q/2}} x^{\alpha} p_{\alpha} =  \sum_{\alpha \in \degmindex{q/2}} x^{\alpha} \cdot \bigg( \sum_{I \in [n]^{q/2} : \mi{I} = \alpha} v_I \bigg)
\]
be the associated polynomial. If $p = 0$, SoS-symmetry of $\sfM$ ensures $v \sfM v^T = 0$. Normalize $v$ so that $\norm{\ell_2}{p} = 1$. 
First, consider another vector $v_m \in [n]^{q/2}$ such that 
\[
(v_m)_I = 
\begin{cases}
\frac{p^{\mi{I}}}{(q/2)!} , & \mbox {if } i_1, \dots, i_{q/2} \mbox{ are all distinct.} \\
 0 & \mbox {otherwise.}
\end{cases}
\]
Then 
\[
\| v_m \|_2^2 \leq \sum_{\alpha \in \degmindex{q/2}} p_{\alpha}^2 / (q/2)! = \frac{1}{(q/2)!},
\]
so $\| v_m \|_2 \leq \frac{2^{O(q)}}{q^{q/4}}$.
Since $\sfA$ is SoS-symmetric, has the minimum eigenvalue at least $-2^{O(q)} \cdot \frac{n^{q/4}}{q^{q/4}}$, and has nonzero entries only on the rows and columns $(i_1, \dots, i_{q/2})$ with all different entries, 
\[
v^T \sfA v = (v_m)^T \sfA (v_m) \geq 
2^{-O(q)} \cdot \frac{n^{q/4}}{q^{3q/4}}.
\]

We finally compute $v^T \sfW v$. 
Let $v_w \in [n]^{q/2}$ be the vector where for each $\alpha \in \degmindex{q/2}$, we choose one $I \in [n]^{q/2}$ arbitrarily and set $(v_w)_I = p_{\alpha}$ (all other $(v_w)_I$'s are $0$). 
By SoS-symmetry of $\sfW$, 
\[
v^T \sfW v = (v_w)^T \sfW (v_w) = p^T \hW p \geq \frac{1}{2},
\]
by~\thmref{identity}. 
Therefore, 
\[
v^T \cdot \sfM \cdot v ~~=~~ 
\frac{1}{c_1} \cdot v^T\cdot \bigg( \frac{1}{c_2} \cdot \frac{q^{3q/4}}{n^{3q/4}} \sfA + \frac{\sfW}{n^{q/2}} \bigg) \cdot v
~~\geq ~~
\frac{1}{c_1} \cdot \bigg( \frac{1}{c_2} \cdot 2^{-O(q)} \cdot \frac{n^{q/4}}{q^{3q/4}} \cdot \frac{q^{3q/4}}{n^{3q/4}} + \frac{1}{2} \cdot \frac{1}{n^{q/2}}  \bigg) 
~~\geq ~~ 0,
\]
by taking $c_2 = 2^{\Theta(q)}$. So $\sfM$ is positive semidefinite, and this finishes the proof
of the lower bound in~\thmref{q-tensor:informal}

\noindent
Thus we obtain, 
\begin{theorem}[Lower bound in \thmref{q-tensor:informal}]
\thmlab{qten:lb}
    For even $q\leq n$, let $\tenAA\in \Re^{[n]^q}$ be a $q$-tensor with 
    i.i.d. $\pm 1$ entries. Then with probability $1-n^{\Omega(1)}$,
    \[
        \frac{\sos{q}{\tenAA(x)}}{\fmax{\tenAA}} ~~\geq ~~
        \pth{\frac{\Omega(n)}{q}}^{q/4-1/2} .
    \]
\end{theorem}

\noindent
As a side note, observe that by applying \lemref{quo:op} and the proof of \lemref{mineig:quo} to the 
SoS-symmetric matrix representation of $f(x) = \tenAA(x)$ (instead of $\sfA$), we obtain a stronger SoS upper bound 
(by polylog factors) for the special case of $d=q$: 
\begin{theorem}[Upper bound in \thmref{q-tensor:informal}]
\thmlab{qten:ub}
    For even $q\leq n$, let $\tenAA\in \Re^{[n]^q}$ be a $q$-tensor with 
    i.i.d. $\pm 1$ entries. Then with probability $1-n^{\Omega(1)}$, 
    \[
        \frac{\sos{q}{\tenAA(x)}}{\fmax{\tenAA}} ~~\leq ~~
        \pth{\frac{O(n)}{q}}^{q/4-1/2} .
    \]
\end{theorem}




\appendix

\section{Upper bounds for Odd Degree Tensors}
\label{app:3ten:ub}
In the interest of clarity, in this section we shall prove \thmref{d-tensor:informal} 
for the special case of $3$-tensors. The proof readily generalizes to the case of all 
odd degree-$d$ tensors.

\subsection{Analysis Overview} 
Let $\tenAA \in \Re^{[n]^3}$ be a $3$-tensor with i.i.d. uniform $\pm 1$ entries. 
Assume $q/4$ is a power of $2$ as this only changes our claims by constants. 
For $\ell\in [n]$ let $\bar{T}_{\ell}$ be an $n\times n$ matrix with i.i.d. uniform 
$\pm 1$ entries, such that we have 
\[
    f(x) := 
    \iprod{\tenAA}{x^{\otimes 3}} = 
    \sum_{\ell\in [n]} x_{\ell} \,(x^T \,\bar{T}_{\ell} \,x) = 
    \sum_{\ell\in [n]} x_{\ell} \,(x^T \,T_{\ell} \,x).
\]
Let $T_{\ell} := (\bar{T}_{\ell}+ \bar{T}_{\ell}^T)/2$.
Following Hopkins et.~al.~\cite{HSS15}, let 
$\mathcal{T}:= \sum_{\ell=1}^{n}T_\ell\otimes T_\ell$. Let 
$E \in \Re^{[n]^2\times [n]^2}$ be the matrix such that $E[(i,i),(j,j)] = 
\mathcal{T}[(i,i),(j,j)]$ for any $i,j\in [n]$ and $E[(i,j),(k,l)] = 0$ otherwise. 
Let $E' \in \Re^{[n]^2\times [n]^2}$ be the matrix such that $E'[(i,j),(i,j)] = 
E[(i,i),(j,j)]+E[(j,j),(i,i)]$ for any $i,j\in [n]$ and $E'[(i,j),(k,l)] = 0$ otherwise. 

Let $T :=  \mathcal{T} - E \in \Re^{[n]^2\times [n]^2}$ and $\matAA := T^{\otimes q/4}$. 
Let $g(x):= (x^{\otimes 2})^T\,\mathcal{T}\,x^{\otimes 2}$ and 
$h(x):= (x^{\otimes 2})^T\,E\,x^{\otimes 2} = (x^{\otimes 2})^T\,E'\,x^{\otimes 2}$.
Let $\PE$ be the pseudo-expectation operator returned by the program above. 

We would like to show that there is some matrix representation $\matBB$ of 
$\matAA$, such that w.h.p. $\max_{\norm{}{y}=1} y^T \matBB y$ is small. To this end, 
consider the following mass shift procedure that we apply to $\matAA$ to get $\matBB$:
\begin{align*}
    \forall I,J\in [n]^{q/2},\quad \matB{I,J} 
    &:= 
    \frac{1}{(q/2)!^2}\sum\limits_{\pi ,\sigma\in \Sym_{q/2}} 
    \matA{\pi(I),\sigma(J)} \\
    &=
    \frac{1}{\cardin{\orbit{I}}\,\cardin{\orbit{J}}}
    \sum\limits_{I'\in \orbit{I} ,J'\in \orbit{J}} \matA{I',J'}
\end{align*}
Below the fold we shall show that $\norm{2}{\matBB}^{4/q} = 
\widetilde{O}(n^{3/2}/\sqrt{q})$ w.h.p. This is sufficient to obtain the desired result 
since we have 
\begin{align*}
    &~~~\quad\norm{2}{\matBB} \Id - B \succeq 0 \\
    &\Rightarrow~
    \norm{2}{\matBB} \|x\|^q -
    \iprod{x^{\otimes q/2}}{\matBB\,x^{\otimes q/2}} \succeq 0 \\
    &\Rightarrow~
    \norm{2}{\matBB} \|x\|^q -
    \iprod{x^{\otimes 2}}{T\,x^{\otimes 2}}^{q/4} \succeq 0 \\
    &\Rightarrow~
    \norm{2}{\matBB} \|x\|^q -
    (g(x)-h(x))^{q/4} \succeq 0 \\
    &\Rightarrow~
    \PEx{(g(x)-h(x))^{q/4}}
    \leq
    \norm{2}{\matBB} \\
    &\Rightarrow~
    \PEx{g(x)-h(x)}
    \leq
    \norm{2}{\matBB}^{4/q} 
    &&\text{(Pseudo-Cauchy-Schwarz)}\\
    &\Rightarrow~
    \PEx{g(x)}
    \leq
    \norm{2}{\matBB}^{4/q} + \PEx{h(x)} \\
    &\Rightarrow~
    \PEx{g(x)}
    \leq
    \norm{2}{\matBB}^{4/q} + 5n
    &&(5n\,\Id - E' \succeq 0) \\
    &\Rightarrow~
    \PEx{g(x)}
    =
    \widetilde{O}(n^{3/2}/\sqrt{q}). 
\end{align*}
\begin{align*}
    \text{Now }\qquad 
    \PEx{f(x)}
    &=
    \PEx{\sum_{\ell\in [n]} x_{\ell} \,(x^T \,T_{\ell} \,x)}
    &&\text{(Following \cite{HSS15})} \\
    &\leq 
    \PEx{\|x\|^2}^{1/2}~
    \PEx{\sum_{\ell\in [n]} (x^T \,T_{\ell} \,x)^2}^{1/2}
    &&\text{(Pseudo-Cauchy-Schwarz)} \\
    &\leq 
    \PEx{\sum_{\ell\in [n]} (x^T \,T_{\ell} \,x)^2}^{1/2}
    &&\text{(Pseudo-Cauchy-Schwarz)} \\
    &\leq 
    \PEx{ \iprod{x^{\otimes 2}}{\mathcal{T}\,x^{\otimes 2}}}^{1/2} \\
    &=
    \PEx{g(x)}^{1/2} 
    =
    \widetilde{O}(n^{3/4}/q^{1/4})\\
\end{align*}

\subsection{Bounding $\|B\|_2^{4/q}$}
For any $I^1,\dots ,I^p\in [n]^{q/2}$ let $e_k$ denote the $k$-th smallest 
element in $\bigcup\limits_{\ell,\,j} \{I^\ell_j\}$ and let 
\[
    \numdist{I^{1},\dots ,I^{p}}:=
        \cardin{\bigcup\limits_{\ell \in [p]} \bigcup\limits_{j \in [n]} \brc{I^{\ell}_{j}}}.
\]
For any $c^1,\dots ,c^s\in [q/2]^p$, let 
\begin{align*}
    &\mathcal{C}(c^1\dots c^s) := \\
    &\quad \brc{(I^1,\dots ,I^p)
    \sep{\numdist{I^1,\dots ,I^p}=s,~\forall k\in [s], 
    \ell\in [p], ~\text{$e_k$ appears $c^k_{\ell}$ times in $I^{\ell}$}}}
\end{align*}

\begin{observation}
\obslab{num:guesses:3}
    \begin{align*}
        \cardin{\brc{(c^1,\dots ,c^s)\in ([q/2]^p)^{s}
        \sep{\mathcal{C}(c^1,\dots ,c^s)\neq \phi}}} 
        &\leq 
        2^{O(pq)} p^{\,pq/2} \\
        \text{if } \mathcal{C}(c^1,\dots ,c^s)\neq \phi \text{ then }
        \cardin{\mathcal{C}(c^1,\dots ,c^s)} 
        &\leq 
        \frac{n^{s}}{s!}\times 
        \frac{((q/2)!)^{p}}{\prod\limits_{\ell\in[p]} 
        c_\ell^{1}!\dots c_\ell^{s}!}
    \end{align*}
\end{observation}

\begin{lemma}
\lemlab{fixed:mult:3}
    Consider any $c^1,\dots ,c^s\in [q/2]^p$ and 
    $(I^{1},\dots ,I^{p})\in \mathcal{C}(c^1,\dots ,c^s)$. We have 
    \[
        \Ex{\matB{I^1,I^2}\matB{I^2,I^3}\dots \matB{I^p,I^1}} 
        \leq 
        2^{O(pq)}n^{pq/8}~\frac{p^{\,5pq/8}}{q^{\,3pq/8}}\,
        \prod_{\ell\in [p]} c_\ell^{1}!\dots c_\ell^{s}!
    \]
\end{lemma}

\begin{proof}
    Consider any $c^1,\dots ,c^s\in [q/2]^p$ and 
    $(I^{1},\dots ,I^{p})\in \mathcal{C}(c^1,\dots ,c^s)$. We have 
    \begin{align}
    \Eqlab{fixed:mult:3}
        &\quad \Ex{\matB{I^1,I^2}\matB{I^2,I^3}\dots \matB{I^p,I^1}} 
        \nonumber\\
        &= 
        \frac{\prod_{\ell} c_\ell^{1}!^{2}\dots c_\ell^{s}!^{2}}
        {((q/2)!)^{2p}}
        \sum\limits_{J^1,K^1\in \orbit{I^1},\dots ,
        J^p,K^p\in \orbit{I^p}} 
        \Ex{\matA{J^1,K^2}\matA{J^2,K^3}\dots \matA{J^p,K^1}} \nonumber\\
        &=
        \frac{\prod_{\ell} c_\ell^{1}!^{2}\dots c_\ell^{s}!^{2}}
        {((q/2)!)^{2p}}\cdot \nonumber\\
        &\quad 
        \sum\limits_{\forall \ell,\, J^\ell,K^\ell\in \orbit{I^\ell}} 
        \Ex{\prod\limits_{\ell\in [p]}
        \matT{J^{\ell}_{1}\,J^{\ell}_{2}\,,\,K^{\ell+1}_{1}K^{\ell+1}_{2}}
        \matT{J^{\ell}_{3}\,J^{\ell}_{4}\,,\,K^{\ell+1}_{3}K^{\ell+1}_{4}}\dots 
        \matT{J^{\ell}_{q/2-1}\,J^{\ell}_{q/2}\,,\,K^{\ell+1}_{q/2-1}\,K^{\ell+1}_{q/2}}} \nonumber\\
        &=
        \frac{\prod_{\ell} c_\ell^{1}!^{2}\dots c_\ell^{s}!^{2}}
        {((q/2)!)^{2p}}
        \sum\limits_{\forall \ell,\, J^\ell,K^\ell\in \orbit{I^\ell}} 
        \Ex{\prod\limits_{\ell\in [p]}\,\,
        \prod_{g\in [q/4]}\,\,
        \sum_{h\in [n]} T_h\pbrcx{J^{\ell}_{2g-1},K^{\ell+1}_{2g-1}} 
        T_h\pbrcx{J^{\ell}_{2g},K^{\ell+1}_{2g}}
        }
        \nonumber\\
        &=
        \frac{\prod_{\ell} c_\ell^{1}!^{2}\dots c_\ell^{s}!^{2}}
        {((q/2)!)^{2p}}
        \sum\limits_{\forall \ell,\, J^\ell,K^\ell\in \orbit{I^\ell}} 
        ~
        \sum_{\forall \ell,g,\, h(\ell,g)\in [n]} 
        \Ex{\prod\limits_{\ell\in [p]}\,\,
        \prod_{g\in [q/4]}
        T_{h(\ell,g)}\pbrcx{J^{\ell}_{2g-1},K^{\ell+1}_{2g-1}} 
        T_{h(\ell,g)}\pbrcx{J^{\ell}_{2g},K^{\ell+1}_{2g}}
        }
        \nonumber\\
        &=
        \frac{\prod_{\ell} c_\ell^{1}!^{2}\dots c_\ell^{s}!^{2}}
        {((q/2)!)^{2p}}
        \sum\limits_{\forall \ell,\, J^\ell,K^\ell\in \orbit{I^\ell}} 
        ~
        \sum_{\biguplus\limits_{u\in [n]} S_u = [p]\times [q/4]} 
        \Ex{\prod\limits_{r\in [n]}\,\,
        \prod_{(\ell,g)\in S_r}
        T_{r}\pbrcx{J^{\ell}_{2g-1},K^{\ell+1}_{2g-1}} 
        T_{r}\pbrcx{J^{\ell}_{2g},K^{\ell+1}_{2g}}
        }
        \nonumber\\
        &= 
        \frac{\prod_{\ell} c_\ell^{1}!^{2}\dots c_\ell^{s}!^{2}}
        {((q/2)!)^{2p}}
        \cardin{
        \mathcal{S}(I^1,\dots ,I^p)
        } \qquad\text{where }\\
        & \mathcal{S}(I^1,\dots ,I^p) := \nonumber\\
        & \mathlarger{\mathlarger{\{}}
        (\Moplus\limits_{\ell\in[p]}(J^\ell,K^\ell)~,~(S_1\dots S_n))
        ~\mathlarger{\mathlarger{\mathlarger{|}}}
        \Moplus\limits_{\ell\in[p]}(J^\ell,K^\ell) \in \prod\limits_{\ell\in[p]} \mathcal{O}^2(I^\ell),~
        (J^{\ell}_{2g-1},K^{\ell+1}_{2g-1})\neq (J^{\ell}_{2g},K^{\ell+1}_{2g}),\nonumber\\
        & 
        \biguplus\limits_{u\in [n]} S_u = [p]\times [q/4],~
        \text{$\forall r\in[n],~\survseqp{S_r}{\tuparg{J}{K}}$ has only even multiplicity elements}
        \mathlarger{\}} \text{,}\nonumber\\
        &\text{ and ~}\survseqp{S_r}{\tuparg{J}{K}} 
        := 
        \bigoplus\limits_{(\ell,g)\in S_r} \pth{ 
        \{J^{\ell}_{2g-1},K^{\ell+1}_{2g-1}\}~,~
        \{J^{\ell}_{2g},K^{\ell+1}_{2g}\}
        } 
        \nonumber
    \end{align}
    Thus it remains to estimate the size of $\survive{I^1,\dots ,I^p}$. We begin with some notation. 
    For a tuple $t$ and a subsequence $t_1$ of $t$, let $t\setminus t_1$ denote the subsequence of 
    elements in $t$ that are not in $t_1$. For a tuple of $2$-sets  
    $t = (\{a_1,b_1\},\dots ,\{a_m,b_m\})$, let $\atomize{t}$ denote the 
    tuple $\pth{a_1,b_1,\dots ,a_m,b_m}$ (we assume 
    $\forall i,\,a_i<b_i$). Observe that ``atomize" is invertible. 
    
    For any $(\tuparg{J}{K},(S_1\dots S_n))\in \survive{I^1,\dots ,I^p}$, observe 
    that $\survseqp{S_r}{\tuparg{J}{K}}$ (which is of length $2|S_r|$) contains a subsequence 
    $\seqI_{S_r}$ of length $|S_r|$, such that $\multiset{\seqI_{S_r}} = 
    \multiset{\survseqp{S_r}{\tuparg{J}{K}}\setminus \seqI_{S_r}}$.  
    Now we know 
    \begin{align}
    \Eqlab{classify:3}
        &\multiset{{\Moplus}_{r}\atomize{\survseqp{S_r}{\tuparg{J}{K}}}} 
        = 
        \bigsqcup_{\ell\in [p]} \multiset{J^\ell \oplus K^\ell} 
        \nonumber\\
        &= \bigsqcup_{\ell\in [p]} \multiset{I^\ell \oplus I^\ell} 
        \nonumber\\
        &\Rightarrow 
        \multiset{{\Moplus}_{r}\atomize{\seqI_{S_r}}}
        =
        \multiset{{\Moplus}_{r}\atomize{\survseqp{S_r}{\tuparg{J}{K}}
        \setminus \seqI_{S_r}}} 
        \nonumber\\
        &= \bigsqcup_{\ell\in [p]} \multiset{I^\ell}. \nonumber\\
        &\text{Thus, } \exists \pi\in \Sym_{pq/2},~s.t.~
        {\Moplus}_{r}\atomize{\seqI_{S_r}} = \pi{\,\pth{I^1\oplus \dots \oplus I^p}} 
    \end{align}
    For tuples $t,t'$, let $\interleave{t}{t'}$ denote the set of all tuples obtained by 
    interleaving the elements in $t$ and $t'$. By \Eqref{classify:3}, we obtain that for 
    any $(\tuparg{J}{K},(S_1\dots S_n))\in \survive{I^1,\dots ,I^p}$, 
    \begin{align}
    \Eqlab{classify2:3}
        \exists \pi\in \Sym_{pq/2},~s.t.~
        \forall r\in [n],~
        \exists \sigma_r\in \Sym_{|S_r|},\,s.t.\quad
        \survseqp{S_r}{\tuparg{J}{K}} \in 
        \interleave{\seqI_{S_r}}{\sigma_r{\,\pth{\seqI_{S_r}}}}& 
        \\
        \text{where } \quad\seqI_{S_r} = 
        \invatomize{\pi{\,\pth{I^1\oplus \dots \oplus I^p}}}& 
        \nonumber
    \end{align}
    For any $j\in [s]$, let $\bar{c}^j := \sum_{\ell\in [p]} c^{j}_{\ell}$.
    Now since $\cardin{\interleave{t}{t'}}\leq 2^{|t|+|t'|}$, by \Eqref{classify2:3} we have 
    that for any $\biguplus_{u\in [n]} S_u = [p]\times [q/4]$,
    \begin{align}
    \Eqlab{partition:3}
        \#(S_1,\dots ,S_n)
        &:= 
        \cardin{\brc{\tuparg{J}{K}\sep{(\tuparg{J}{K},(S_1\dots S_n))\in 
        \survive{I^1,\dots ,I^p}}}} \nonumber\\
        &\leq 2^{pq/4}\cardin{\Sym_{|S_1|}}\times \dots 
        \times \cardin{\Sym_{|S_n|}}\times
        \cardin{\orbit{I^1\oplus \dots \oplus I^p}} \nonumber\\
        &\leq 
        2^{pq/4}\,|S_1|!\dots |S_n|!\, \frac{(pq/2)!}{\bar{c}^1!\dots \bar{c}^s!}
    \end{align}
    
    For any $(\tuparg{J}{K},(S_1\dots S_n))\in \survive{I^1,\dots ,I^p}$, observe that the even 
    multiplicity condition combined with the condition that $(J^{\ell}_{2g-1},K^{\ell+1}_{2g-1})\neq 
    (J^{\ell}_{2g},K^{\ell+1}_{2g})$, imply that for each $r\in [n]$, $|S_r|\neq 1$. Thus every 
    non-empty $S_r$ has size at least $2$, implying that the number of non-empty sets in 
    $S_1,\dots ,S_n$ is at most $pq/8$. Thus we have,
    \begin{align*}
        &\cardin{\survive{I^1,\dots ,I^p}} \\
        &=
        \sum_{U\subseteq [n], |U|\leq pq/8}~
        \sum_{\biguplus\limits_{u\in U} S_u = [p]\times [q/4]}
        \#(S_1,\dots ,S_n) 
        &&(S_r = \phi \text{ if } r\not\in U)\\
        &=
        \sum_{U\subseteq [n], |U|\leq pq/8}~
        \sum_{\sum_{u\in U}s_u = pq/4}~
        \sum_{|S_r| = s_r,\,\biguplus\limits_{u\in U} S_u = [p]\times [q/4]}
        \#(S_1,\dots ,S_n) 
        &&(s_r = 0, S_r = \phi \text{ if } r\not\in U)\\
        &\leq 
        \sum_{U\subseteq [n], |U|\leq pq/8}~
        \sum_{\sum_{u\in U}s_u = pq/4}~
        \#(S_1,\dots ,S_n) \binom{pq/4}{s_1}\dots \binom{pq/4}{s_n}
        &&(s_r = 0, S_r = \phi \text{ if } r\not\in U)\\
        &\leq 
        \sum_{U\subseteq [n], |U|\leq pq/8}~
        \sum_{\sum_{u\in U}s_u = pq/4}~
        2^{pq/4}\,s_1!\dots s_n!\, \frac{(pq/2)!}{\bar{c}^1!\dots \bar{c}^s!} 
        \binom{pq/4}{s_1}\dots \binom{pq/4}{s_n}
        &&(\text{by \Eqref{partition:3}})\\
        &\leq 
        \sum_{U\subseteq [n], |U|\leq pq/8}~
        \sum_{\sum_{u\in U}s_u = pq/4}~
        2^{O(pq)}\,\frac{(pq/2)!}{\bar{c}^1!\dots \bar{c}^s!} 
        (pq)^{s_1+\dots +s_n}
        &&(s_r = 0, S_r = \phi \text{ if } r\not\in U)\\
        &\leq 
        \sum_{U\subseteq [n], |U|\leq pq/8}~
        2^{O(pq+|U|)}\,\frac{(pq/2)!}{\bar{c}^1!\dots \bar{c}^s!} 
        (pq)^{pq/4} \\
        &\leq 
        \sum_{\bar{u}\in [pq/8]}~
        \sum_{U\subseteq [n], |U|= \bar{u}}
        2^{O(pq+|U|)}\,\frac{(pq/2)!}{\bar{c}^1!\dots \bar{c}^s!} 
        \,(pq)^{pq/4} \\
        &\leq 
        \sum_{\bar{u}\in [pq/8]}
        2^{O(pq)}\,\binom{n}{\bar{u}}\,
        \frac{(pq/2)!}{\bar{c}^1!\dots \bar{c}^s!} 
        (pq)^{pq/4} \\
        &\leq 
        \sum_{\bar{u}\in [pq/8]}
        2^{O(pq+\bar{u})}\,\frac{n^{\bar{u}}}{\bar{u}^{\bar{u}}}\,
        \frac{(pq/2)!}{\bar{c}^1!\dots \bar{c}^s!} 
        (pq)^{pq/4} \\
        &\leq 
        \sum_{\bar{u}\in [pq/8]}
        2^{O(pq)}\,(npq)^{pq/8}\,
        \frac{(pq/2)!}{\bar{c}^1!\dots \bar{c}^s!}~
        \frac{(pq)^{pq/8}}{n^{pq/8-\bar{u}}\,\bar{u}^{\bar{u}}}\\
        &\leq 
        \sum_{\bar{u}\in [pq/8]}
        2^{O(pq)}\,(npq)^{pq/8}\,
        \frac{(pq/2)!}{\bar{c}^1!\dots \bar{c}^s!}~
        \frac{(pq)^{\bar{u}}}{\bar{u}^{\bar{u}}}
        &&(\text{ since } pq<n)\\
        &\leq 
        \sum_{\bar{u}\in [pq/8]}
        2^{O(pq)}\,(npq)^{pq/8}\,
        \frac{(pq/2)!}{\bar{c}^1!\dots \bar{c}^s!} \\
        &\leq 
        2^{O(pq)}\,(npq)^{pq/8}\,
        \frac{(pq/2)!}{\bar{c}^1!\dots \bar{c}^s!} \\ \\
        \Rightarrow~ &
        \Ex{\matB{I^1,I^2}\matB{I^2,I^3}\dots \matB{I^p,I^1}} \\
        &\leq 
        2^{pq/4}\,(npq)^{pq/8}\, \frac{(pq/2)!}{\bar{c}^1!\dots \bar{c}^s!} 
        \frac{\prod_{\ell} c_\ell^{1}!^{2}\dots c_\ell^{s}!^{2}}{((q/2)!)^{2p}}
        &&\text{(by \Eqref{fixed:mult:3})}\\
        &\leq 
        2^{pq/4}\,(npq)^{pq/8}\, (pq/2)!\, 
        \frac{\prod_{\ell} c_\ell^{1}!\dots c_\ell^{s}!}{((q/2)!)^{2p}} \\
        &= 
        2^{O(pq)}\,n^{pq/8}\,\frac{p^{\,5pq/8}}{q^{\,3pq/8}}\,
        \prod_{\ell\in [p]} c_\ell^{1}!\dots c_\ell^{s}!
    \end{align*}
\end{proof}

\begin{lemma}
\lemlab{num:dist:3}
    For all $i^{1},\dots ,i^{p}\in [n]^{q/2}$, we have 
    \begin{align*}
        &(1)\quad \Ex{\matB{I^1,I^2}\matB{I^2,I^3}\dots \matB{I^p,I^1}} \geq 0 \\
        &(2)\quad \Ex{\matB{I^1,I^2}\matB{I^2,I^3}\dots \matB{I^p,I^1}} \neq 0 
        \quad \Rightarrow \quad \numdist{I^{1},\dots ,I^{p}} \leq \frac{pq}{4}+\frac{q}{2}
    \end{align*}
\end{lemma}

\begin{proof}
    The first claim follows immediately on noting that one is taking expectation of a 
    polynomial of independent centered random variables with all coefficients positive. 
    
    For the second claim, note that $\Ex{\matB{I^1,I^2}\matB{I^2,I^3}\dots \matB{I^p,I^1}} \neq 0$ 
    implies that $\survive{I^1,\dots ,I^p}\neq \phi$. Therefore there exists $\tuparg{J}{K}$ 
    (where $J^{\ell},K^{\ell}\in \orbit{I^{\ell}}$) and 
    $\biguplus_{u\in [n]} S_u = [p]\times [q/4]$ such that every element in 
    $\Moplus_{r\in [n]}\survseqp{S_r}{\tuparg{J}{K}}$ has even multiplicity. 
    The rest of the proof follows from the same ideas as in the proof of \lemref{num:dist}.
\end{proof}

\begin{lemma}
    \[\norm{2}{\matBB}^{4/q}\leq \frac{n^{3/2}\log^5 n}{\sqrt{q}} \quad w.h.p.\]
\end{lemma}

\begin{proof}
We proceed by trace method. (Note that since $T$ is symmetric, so are $A$ and $B$). 
\begin{align*}
    &\quad \Ex{\Tr{\matBB^{p}}} \\
    &= 
    \sum\limits_{I^1,\dots ,I^p \in [n]^{q/2}}
    \Ex{\matB{I^1,I^2}\matB{I^2,I^3}\dots \matB{I^p,I^1}} \\
    &= 
    \sum\limits_{s\in [pq/4 + q/2]} 
    ~\sum\limits_{\numdist{i^1,\dots ,i^p}=s}
    \Ex{\matB{I^1,I^2}\matB{I^2,I^3}\dots \matB{I^p,I^1}} 
    &&\text{by \lemref{num:dist:3}}\\
    &= 
    \sum\limits_{s\in [pq/4 + q/2]} 
    ~\sum\limits_{c^1,\dots ,c^s\in [q/2]^p} 
    ~\sum\limits_{(I^1,\dots ,I^p) \in \mathcal{C}(c^1\dots c^s)}
    \Ex{\matB{I^1,I^2}\matB{I^2,I^3}\dots \matB{I^p,I^1}} \\
    &= 
    \sum\limits_{s\in [pq/4 + q/2]} 
    ~\sum\limits_{c^1,\dots ,c^s\in [q/2]^p} 
    ~\sum\limits_{(I^1,\dots ,I^p) \in \mathcal{C}(c^1\dots c^s)}
    \Ex{\matB{I^1,I^2}\matB{I^2,I^3}\dots \matB{I^p,I^1}} \\
    &\leq 
    \sum\limits_{s\in [pq/4 + q/2]} 
    ~\sum\limits_{c^1,\dots ,c^s\in [q/2]^p} 
    ~\sum\limits_{(I^1,\dots ,I^p) \in \mathcal{C}(c^1\dots c^s)}
    2^{O(pq)}\,n^{pq/8}\,\frac{p^{\,5pq/8}}{q^{\,3pq/8}}\,
    \prod_{\ell\in [p]} c_\ell^{1}!\dots c_\ell^{s}! 
    &&\text{by \lemref{fixed:mult:3}} \\
    &\leq 
    \sum\limits_{s\in [pq/4 + q/2]} 
    2^{O(pq)}\,\frac{n^{s+pq/8}}{s!}\,p^{\,9pq/8}q^{\,pq/8}
    &&\text{by \obsref{num:guesses:3}} \\
    &\leq 
    \sum\limits_{s\in [pq/4 + q/2]} 
    2^{O(pq)}\,\frac{n^{\,pq/4+q/2+pq/8}}
    {s!\,q^{\,pq/4+q/2-s}}\,p^{\,9pq/8}q^{\,pq/8}
    &&(\text{since } q\leq n) \\
    &\leq 
    \sum\limits_{s\in [pq/4 + q/2]} 
    2^{O(pq)}\,\frac{n^{\,3pq/8+q/2}\,p^{\,9pq/8}}{q^{\,pq/8}}
    ~\leq ~ 
    2^{O(pq)}\,\frac{n^{\,3pq/8+q/2}\,p^{\,9pq/8}}{q^{\,pq/8}}.
\end{align*}
Choose $p$ to be even and let $p=\Theta(\log n)$. 
Now 
\[ \Prob{\norm{2}{\matBB}^{4/q}\geq n^{3/2}\log^5 n/\sqrt{q}} \leq 
\Prob{\Tr{\matBB^{p}} \geq n^{\Omega(1)}\Ex{\Tr{\matBB^{p}}}} \ . \] 
Applying Markov inequality completes the proof. 
\end{proof}

Thus we obtain
\begin{theorem}
\thmlab{3ten:ub}
    Let $\tenAA\in \Re^{[n]^3}$ be a $3$-tensor with i.i.d. 
    $\pm 1$ entries. Then for any even $q$ such that $q\leq n$, 
    we have that with probability $1-n^{\Omega(1)}$, 
    \[
        \frac{\sos{q}{\tenAA(x)}}{\fmax{\tenAA}} ~~\leq ~~
        \pth{\frac{\widetilde{O}(n)}{q}}^{1/4} .
    \]
\end{theorem}

\bibliographystyle{alpha}
\bibliography{polynomials}

\end{document}